\documentclass[11pt,3p]{scrartcl}
 
\usepackage[hidelinks]{hyperref}

\usepackage[x11names]{xcolor}%

\usepackage{amsmath,amssymb,amsfonts}%

\usepackage{chngcntr}
 
\counterwithin{figure}{section}
\counterwithin{table}{section}

\usepackage{textcomp}%
\usepackage{manyfoot}%
\usepackage{booktabs}%
\usepackage{listings}%

\usepackage{hyperref}
\usepackage[centerdot]{mathtools}
\usepackage[T1]{fontenc}
\usepackage[utf8]{inputenc}
\usepackage{lmodern}
\usepackage[english]{babel}

\usepackage{amsthm}
\usepackage[]{mathtools}
\usepackage{dsfont}
\newcommand{\R}{\mathds{R}}
\newcommand{\Z}{\mathds{Z}}

\newcommand{\X}{\mathcal{X}}
\newcommand{\Y}{\mathcal{Y}}
\newcommand{\Yn}{\mathcal{Y}_N}

\usepackage{natbib}
\AtEndOfPackage{%
   \renewcommand\@openbib@code{%
      \advance\leftmargin\bibindent
      \itemindent -\bibindent
      \listparindent \itemindent
      \parsep \z@
      }%
   \renewcommand\newblock{\par}}%

\AtEndOfClass{\RequirePackage{natbib}%
\setlength{\bibhang}{\parindent}%

\def\@biblabel#1{#1.}%
\bibpunct{(}{)}{;}{a}{}{,}}

\usepackage{pifont}%

\DeclareMathOperator{\conv}{conv}

\usepackage[ruled,noend]{algorithm2e}
\SetKw{Continue}{continue}
\SetKw{Null}{null}
\SetKw{Break}{break}

\SetCommentSty{mycommfont}

\usepackage[normalem]{ulem}
\usepackage{float}
\usepackage{subcaption}
\usepackage{authblk}
\usepackage{rotating}

\usepackage{tikz}
\usepackage{tikz,tikz-3dplot}
\tikzset{
    cross/.pic = {
    \draw[rotate = 45] (-#1,0) -- (#1,0);
    \draw[rotate = 45] (0,-#1) -- (0, #1);
    }
}

\usetikzlibrary{shapes.misc}
\usetikzlibrary{shapes.geometric}
\usetikzlibrary{patterns}
\usetikzlibrary{graphs,quotes}

\usepackage{pgfplots}
 \pgfplotsset{compat=1.16,
    every axis/.append style={
        axis lines=center,
        xlabel style={anchor=south west},
        ylabel style={anchor=south west},
        zlabel style={anchor=south west},
        tick align=outside,
        every tick label/.append style={font=\scriptsize}}
}
 \usepgfplotslibrary{patchplots}
\tdplotsetmaincoords{60}{120}
\usepgfplotslibrary{fillbetween}

\usepackage[]{cleveref} 

\newtheoremstyle{theorem}
{15pt}
{5pt}
{\itshape}
{}
{\bfseries}
{}
{0.5em}
{}
\theoremstyle{theorem}

\usepackage{xpatch}

\xpatchcmd{\proof}{\itshape}{\prooflabelfont}{}{}
\newcommand{\prooflabelfont}{\bfseries}

\newtheorem{theorem}{Theorem}

\newtheorem{definition}[theorem]{Definition}

\title{{Representing the Non-dominated Set 
of Multi-objective Network Problems 
by Supported Non-dominated Points}}

\author[1,2]{David K\"onen}
\author[1,3]{Lara L\"ohken}
\author[1,4]{Michael Stiglmayr}
\affil[1]{%
	University of Wuppertal\\
	School of Mathematics and Natural Sciences\\
	Optimization Group\\
	Gaußstraße 20, 42103 Wuppertal, Germany\\ 
}

\affil[2]{Corresponding author, E-Mail:~\href{mailto:koenen@uni-wuppertal.de}{koenen@uni-wuppertal.de}, ORCID-ID: 0000-0003-1747-8791}
\affil[3]{E-Mail:~\href{mailto:loeken@uni-wuppertal.de}{loehken@uni-wuppertal.de}, ORCID-ID: 0009-0003-9205-6641}
\affil[4]{E-Mail:~\href{mailto:stiglmayr@uni-wuppertal.de}{stiglmayr@uni-wuppertal.de}, ORCID-ID: 0000-0003-0926-1584 }

\date{}
\begin{document}

\maketitle

	\begin{abstract}\small
    In multi-objective combinatorial optimization, unsupported non-dominated points typically outnumber supported points and are often significantly more challenging to compute. Recent studies show that extreme supported non-dominated points provide high-quality representations of the non-dominated set for certain binary problems. We demonstrate that this observation does not generalize to capacitated network optimization problems: representation quality decreases with increasing arc capacities, whereas supported non-dominated points consistently provide high-quality representations with respect to several quality indicators. However, supported point sets may still be too large in practical applications, where only a small, fixed number of alternatives is typically desired. Selecting fixed-size representations from the non-dominated set requires its computationally expensive generation and thus diminishes the computational advantages that representations are intended to provide. We therefore suggest the (extreme) supported points as alternative candidate sets in subset selection problems. Our numerical results show that restricting the candidate set to supported non-dominated points yields fixed-size representations of nearly the same quality as those selected from the complete non-dominated set. Overall, supported non-dominated points serve both as high-quality representations and as reasonable candidate sets for subset selection.
\end{abstract}

\par\vskip\baselineskip\noindent

\section{Introduction}\label{sec:Intro}

Network optimization problems are fundamental optimization problems, arising when a commodity moves through an underlying network, creating a flow of resources or information. Nowadays, such networks appear in various forms in our daily lives. The textbook by~\cite{Ahuja1993} presents over 150 applications of network optimization problems across various fields, such as engineering, management, and other scientific domains. 
While single-objective variants have been studied since the early 1950s~\citep[see][and the references given therein]{Ahuja1993,bertsekas98network}, real-world applications, however, often involve conflicting objectives.

These problems belong to the general class of multi-objective integer linear problems (MOILP), defined as

\begin{equation}\label{eq:MOCO}
        \tag{MOILP}
     \min_{x\in\X}  z(x)=\left(z_1(x),\ldots,z_p(x)\right)^{\top} = Cx,     
\end{equation}
with the cost matrix $C \in \mathds{Z}^{p \times n}$ containing the rows $c^k$ of coefficients of $p\geq 2$
linear objective functions $z_k(x)=c^k x$ for  $k \in \{1, \ldots,p \}$ and the feasible set $\X \coloneqq \{ x  \in  \Z^n \colon Ax = b \}$, where $A \in \mathds{Z}^{m\times n}, b \in\mathds{Z}^m$ describe the $m$ constraints.  If $x\in \{0,1\}^n$ these problems are also referred to multi-objective combinatorial optimization (MOCO) problems. The set of feasible outcome vectors in the objective space is denoted by $\Y \coloneqq   \{C\, x \colon x\in \X \}$, where $\X$ is the set of all feasible solutions.

In multi-objective optimization, it is usually assumed that the objective functions are conflicting, which implies that no solution minimizes all objectives simultaneously. In this situation, we use the \emph{Pareto concept of optimality}, based on the component-wise orders in $\R^p$. Let $y^1, y^2 \in \mathds{R}^p$ then
\begin{align*}
& y^1 \leqq y^2 \iff y_k^1 \leq y^2_k \quad \forall \; k \in \{1,\dots ,p\} \\
& y^1 \leqslant y^2 \iff y^1 \leqq y^2 \quad \text {and } \quad y^1 \neq y^2\\
& y^1 < y^2 \iff y_k^1 < y^2_k \quad \forall \; k \in \{1,\dots ,p\}  
\end{align*}

Accordingly, we define the non-negative orthant denoted by $\mathds{R}_{\geqq}^p:=\left\{y \in \mathds{R}^p: y \geqq 0\right\}$, and analogously by $\R_{>}^p$ and $\R_{\geqslant}^p$.

\begin{definition}
    A feasible solution $x^{*} \in \X$ is \emph{efficient} if there is no feasible solution $x \in \X$ such that $f(x) \leqslant f\left(x^{*}\right)$. If $x^{*}$ is efficient, $ f\left(x^{*}\right)$ is called a \emph{non-dominated point}. If $f(x) \leqslant f\left(x^{\prime}\right)$ holds for two feasible solutions $x, x^{\prime} \in \X$, then $x$ is said to \emph{dominate} $x^{\prime}$, and $f(x)$ \emph{dominates} $f\left(x^{\prime}\right)$. 
    The set of efficient solutions is denoted by $\X_E\subseteq \X$ and the set of non-dominated points by $\mathcal{Y}_N\subseteq \Y$.
    \end{definition}

Assuming that $\Yn$ is bounded, we can define lower and upper component-wise bounds for all non-dominated points by the \emph{ideal point} $y^I$ with $y_k^I \coloneqq \min_{y \in \Y} y_k $ and the \emph{Nadir point} $y^N$ with $y_k^N \coloneqq \max_{y \in \Yn} y_k$ for $k \in \{1, \dots, p\}$.

\paragraph{Scalarization and Supported Solutions} 
The most common solutions approaches for multiobjective optimization problems rely on scalarizations. Thereby, the multi-objective problem is transformed into a sequence of single-objective optimization problems whose solutions correspond to efficient solutions of the original problem. For a comprehensive introduction, we refer e.\,g.\ to~\cite{ehrgott2005multicriteria}. 

The most straightforward scalarization approach is the \emph{weighted sum method}, where a convex combination of the objectives is minimized for varying weighting vectors. Let $\|x\|_{1}$ denote the $1$-norm of $x \in \R^{d}$, i.\,e.,  $\|x\|_{1}\coloneqq \sum_{i=1}^{d}\left|x_{i}\right|$. The set of \emph{normalized weighting vectors} is defined as 
$\Lambda_p=\{\lambda \in \R^{p}_{>} \colon \|\lambda\|_{1} = 1 \}$, and 
$\Lambda_p^0=\{\lambda \in \R^{p}_{\geq} \colon \|\lambda\|_{1} = 1 \}$ 
if weights equal to zero are included.
Then, the associated \emph{weighted-sum scalarization} with $\lambda \in \Lambda_p$ or $\lambda \in \Lambda_p^0$ is defined as the parametric program 
\begin{equation}
    \tag{$P_{\lambda}$}\label{eq:wsm}
    \min_{x\in \X} \lambda^\top z(x).
\end{equation}

An efficient solution is called \emph{a supported efficient solution} if it is an optimal solution to the weighted sum scalarization~\ref{eq:wsm} for  $\lambda\in \Lambda_{d}$, i.\,e., an optimal solution to a single-objective weighted sum problem where the weights are strictly positive. Its image is called \emph{supported non-dominated point}; we use the notation $\X_{S}$ and $\Y_{S}$. 
It is known that MOILPs and, more generally, nonconvex multi-objective problems can contain \emph{unsupported non-dominated points} which correspond to efficient solutions for which there is no $\lambda \in \Lambda_{p}^0$ such that they are optimal for~\ref{eq:wsm}. Unsupported non-dominated points are located in the interior of the \emph{convexified upper image} $\mathcal{P}\coloneqq  \operatorname{conv}(\Y+ \smash{\R_\geqq^p})$, i.e., the convex hull of the upper image. There may also exist \emph{weakly supported efficient solutions} that are optimal solutions of $P_{\lambda}$ for  $\lambda\in \Lambda_{p}^0$, i.\,e., an optimal solution to a weighted sum single-objective problem with non-negative weights. Their images in the objective space are weakly \emph{supported non-dominated points}. The distinction between supported and weakly supported solutions is relatively new, first proposed in~\cite{koenen2025}, to provide a consistent definition of supportedness. \emph{Extreme supported solutions}  are those solutions whose images lie on the vertex set of the {convexified} upper image. Their images are called \emph{extreme supported non-dominated vectors}, denoted by $\Y_{ES}$ and $\X_{ES}$ for their preimages.

\paragraph{Motivation and Related Literature}
MOCO problems generally belong to the class of computationally intractable problems, i.\,e., the number of non-dominated points can grow exponentially with the size of the instance. This results in high computational effort, particularly for high-dimensional problems or large instances.
To address this challenge, researchers explore representation techniques designed to extract a small subset of $\mathcal{Y}_N$ to present to a decision maker. There are several concepts to measure the quality of such a representation (in terms of size, non-redundancy, and coverage), see, 
e.\,g.,~\cite{Sayin2000} and~\cite{HUMBERTROPERS2025}. Particularly for MOCO problems, there are various publications focusing on representation or approximation methods based on different quality indicators (see, e.\,g.,~\citealt{hamacher07,vaz15representation, rios2021, mesqui2022, Serpil2024, Kirlik2025, HUMBERTROPERS2025}). For a comprehensive survey of representations in multi-objective optimization and related methods, we refer to~\cite{Faulkenberg2010} and~\cite{Herzel2021}.  

A recent study by~\citet{Serpil2024} demonstrates that the set of supported non-dominated points yield high-quality representations for multi-objective discrete optimization problems. It is further shown that $|\Y_{S}|$ and $|\Y_{ES}|$ coincide for all evaluated knapsack problem instances and are nearly identical for all (binary) assignment problems. Thus, the extreme supported points alone are sufficient to provide a high-quality representation for these problem classes.
This result can be of large practical importance, as (extreme) supported non-dominated points are usually calculated relatively efficiently compared to unsupported points, which, moreover, typically outnumber them by far (see, e.\,g.,~\citealt{Visee1998,raith09, figueira17easy}). 

Several studies focus on identifying or analyzing the supported or extreme supported point set across various network flow or graph optimization problems, including bi- and multi-objective integer network flow problems~\citep{eusebio14,Church2015,raith17,koenen2023outputpolynomial,koenen2023outputsensitive}, bi- and multi-objective minimum spanning trees~\citep{sourd2006multi,silva07note, Ruzika2009, Correia2021}, bi- and multi-objective shortest path problems~\citep{EDWIN99,SEDENONODA15,church14}, (bi-objective) cable-trench problem~\citep{Vasko2002, Vasko2015, Loehken2025} and bi- and multi-objective assignment problems~\citep{tuyttens2000performance,gandibleux2003use,PRZYBYLSKI2010149}. Additional work on general MOCO problems can be found in~\cite{gandibleux01},~\cite{jesus2015implicit}, and \cite{Serpil2024} among others. However, there is no research evaluating the \emph{representation quality} of the set of supported non-dominated points for network optimization problems. 

This leads to our first central research question: \emph{Does the observation that extreme supported points provide a sufficiently high-quality representation generalize to network optimization problems, particularly those with higher arc capacities?}

\paragraph{Contribution} 
In this work, we demonstrate, as a first contribution, that the observations in~\citet{Serpil2024} do not generalize to all network optimization problems. 
Indeed, for binary network optimization problems, e.\,g., the multi-objective minimum spanning tree problem, we obtain results similar to those reported in~\cite{Serpil2024}, namely that $|\Y_{S}|$ and $|\Y_{ES}|$ are nearly identical, and $\Y_{ES}$ (alone) provides a high-quality representation. However, this behavior does not extend to network problems with higher arc capacities, such as the multi-objective minimum cost flow problem. The number of supported non-dominated points significantly exceeds the number of extreme supported points for such capacitated problems, and the set $\Y_{S}$ constitutes a representation of much higher quality with respect to various quality indicators than the set $\Y_{ES}$.

Despite $\Y_{S}$ or $\Y_{ES}$ offering high-quality representations, their cardinalities may still be too large, as decision makers typically expect a small and manageable subset of alternatives (e.\,g., 5--20 points). 
In such cases, a posteriori methods are commonly used to derive fixed-size representations by solving a subset selection problem with respect to a specific quality indicator. These approaches are typically based on the explicit computation of the entire non-dominated set, from which a representative subset is chosen. However, this largely eliminates the computational advantages that representations are intended to provide. To bypass this bottleneck, we propose restricting the candidate pool to the easier-to-compute sets $\Y_{S}$ or $\Y_{ES}$. This motivates our second research question: \emph{How much representation quality is lost when the candidate pool for extracting a fixed-size representation of size $k$ is restricted to $\Y_{S}$ or $\Y_{ES}$ instead of the full non-dominated set $\Y_{N}$?}

As a second contribution, we address this question by analyzing the loss in representation quality with respect to different quality indicators. We demonstrate that the representation obtained from the supported set is nearly as good as one selected from the complete non-dominated set, offering a massive reduction in computational effort with a minor loss in quality. 

We present numerical experiments in which various quality metrics, such as coverage error, hypervolume ratio, and $\varepsilon$-indicator, are used to assess the quality of representational subsets, thereby demonstrating our findings. All test instances and implementations are publicly available in the associated repository~\citep{koenen2026supprepresentations} to ensure reproducibility and to facilitate further research, which constitutes our third contribution, particularly with regard to open science.

This paper is organized as follows. In~\Cref{sec:NetworkProb}, we introduce the considered network optimization problems and discuss related literature.~\Cref{sec:Rep} presents the mathematical definitions of different quality measures for representational subsets. In~\Cref{sec:RepQualSuppNonPoint}, we provide a numerical evaluation of the quality of the set of (extreme) supported non-dominated points as representations for the considered graph problems, followed by the analysis of fixed-size representations in~\Cref{sec:FixedSize}. The results are summarized and concluded in~\Cref{chapt:concl}.

\section{Network Optimization Problems} 
\label{sec:NetworkProb}

 We distinguish two classes of network optimization problems which we investigate separately:

\begin{description}
    \item[Binary and Low-Capacity Structures] Problems where feasible solutions are defined by binary decision variables (e.\,g., selecting edges for a tree or path) or where arc capacities are restricted to one or other small values. For these structures, we hypothesize that the set of extreme supported non-dominated points ($\Y_{ES}$) is sufficient for a high-quality representation of $\Y_{N}$.
    \item[Capacitated Flow Structures] General flow problems where arc capacities allow for integer flows significantly greater than one. Here, we hypothesize that non-extreme supported points play a crucial role in high-quality representation.
\end{description}

\subsection{Binary and Low-Capacity Problems}\label{subsec:BinProb}

\paragraph{Minimum Spanning Tree Problem}
The minimum spanning tree problem asks for a spanning tree of a weighted, connected, undirected graph that minimizes the sum of its edge weights. For a comprehensive overview, we refer to the textbook by~\citet{Ahuja1993}. For the multi-objective variant, we refer to~\citet{ehrgott00survey}.

\begin{definition}[Multi-objective Minimum Spanning Tree]
    Given a graph $G=(V,E)$ and a cost function $c\colon E \to \Z^p$ that assigns a cost vector $c(e) = ( c_1(e), \dots, c_p(e))^\top$ to each edge $e \in E$. Then the \emph{multi-objective minimum spanning tree} problem can be written as
    \begin{equation}\label{eq:MOMST}
        \tag{MOMST}
     \min_{T\in\mathcal{T}}  c(T)=\left(c_1(T),\ldots,c_p(T)\right)^{\top}   
\end{equation}
    where $\mathcal{T}$ is the set of all spanning trees of $G$, and $c_i(T) = \sum_{e\in T} c_i(e)$ for $i=1,\ldots,p$ are the $p$ linear objective functions.
\end{definition}

The~\ref{eq:MOMST} is known to be intractable~\citep{hamacher1994spanning}. Moreover, the associated decision problem is $\mathbf{NP}$-hard, even in the case of only two objectives~\citep{Camerini1984}. However, the number of extreme supported non-dominated points of~\ref{eq:MOMST} is polynomially bounded~\citep{Seipp2013}.

\paragraph{Shortest Path Problem}
The shortest path problem aims to find a path between two predefined nodes in a weighted, directed graph that minimizes the sum of the arc weights along the path. For a comprehensive overview, we refer again to the textbook by~\citet{Ahuja1993}, and for the multi-objective variant, we refer to~\citet{ehrgott00survey}.

\begin{definition}[Multi-objective Shortest Path Problem]

Given a directed graph $D = (V, A)$, a cost function $c\colon A \to \mathbb{Z}^p_{\geq}$ that assigns a cost vector $c(a) = ( c_1(a), \dots, c_p(a))^\top $ to each arc $a \in A$, a source node \(s \in V\), and a target node \(t \in V\), the  \emph{multi-objective shortest path problem} is defined by

\begin{equation}\label{eq:MOSP}
        \tag{MOSP}
     \min_{P\in\mathcal{P}_{s,t}} c(P)=\left(c_1(P),\ldots,c_p(P)\right)^{\top}     
\end{equation}
    where \(\mathcal{P}_{s,t}\) is the set of all dipaths from \(s\) to \(t\) in \(D\), and $c_i(P) = \sum_{a\in P} c_i(a)$ for all $i=1,\ldots,p$. 
    
\end{definition}

The decision version of the multi-objective shortest path problem is $\mathbf{NP}$-hard~\citep{serafini87some}, and moreover,~\ref{eq:MOSP} is intractable~\citep{Hansen1979}. Unlike ~\ref{eq:MOMST}, even the number of extreme supported non-dominated points can grow exponentially with the size of an instance, see~\cite{Ruhe1988}. Moreover, all of these complexity results also hold for the bi-objective case. 

\paragraph{Cable-trench Problem} The cable-trench problem and its bi-objective formulation are a combination of the shortest path and the minimum spanning tree problem. The goal is to find a spanning tree of a weighted, connected graph that simultaneously minimizes its total cost and the total path costs from a predefined source node to all other nodes in the graph. The cable-trench problem is introduced as a single-objective optimization problem, whose objective function is a parameterized sum of spanning tree and shortest path objectives~\citep{Vasko2002, Vasko2015}. A bi-objective approach to the cable-trench problem can be formulated by considering spanning tree and shortest path objectives as two individual objective functions that are minimized separately. Indeed, the cable-trench problem in its original version corresponds to a weighted sum scalarization of its bi-objective version. We refer to~\cite{Loehken2025} for further reading.

\begin{definition}[Bi-objective Cable-trench Problem]
Given a graph $G=(V,E)$ with specified source node  $v_0 \in V$ and a cost function $c\colon E \rightarrow \R$ that assigns a cost coefficient $c(e)$ to each edge. The \emph{bi-objective cable-trench problem} is given by
\begin{equation}\label{eq:BCTP}
        \tag{BCTP}
     \min_{T\in\mathcal{T}}  c(T)=\left(c_{\gamma}(T),c_{\tau}(T)\right)^{\top}      
\end{equation}

where $\mathcal{T}$ denotes the set of all spanning trees of $G$, $ c_{\gamma}(T)$ the total path cost of all paths from $v_0$ to all other nodes of $V$ in $T$ and $c_{\tau}(T)$ the total cost of the spanning tree $T$.
    
\end{definition}

Assuming non-proportional cost coefficients with respect to the two objective functions allows us to define the generalized (bi-objective) cable-trench problem~\citep{Vasko2002, Loehken2025}. 

\begin{definition}[Generalized Bi-objective Cable-trench Problem]

Let $c\colon E \rightarrow \R^2$ be a cost function that assigns two cost coefficients $\alpha(e)$ and $\beta(e)$ to each edge, i.\,e., $c(e)=(\alpha(e), \beta(e))^\top$. The \emph{generalized bi-objective cable-trench problem} can be written as 

\begin{equation}\label{eq:BGCTP}
        \tag{BGCTP}
     \min_{T\in\mathcal{T}}  c(T)=\left(\alpha_{\gamma}(T),\beta_{\tau}(T)\right)^{\top}   
\end{equation} 

where $\alpha_{\gamma}(T)$ denotes the total path cost, determined by edge costs $\alpha(e)$, and $\beta_{\tau}(T)$ denotes the total cost of the spanning tree $T$, determined by the cost coefficients $\beta(e)$.
\end{definition}

The bi-objective (generalized) cable-trench problem is intractable, and the associated decision problem is $\mathbf{NP}$-hard~\citep{Loehken2025}. 

\subsection{Capacitated Flow Problems}\label{subsec:CapFlowProb}

\paragraph{Minimum Cost Flow Problem}
In the minimum cost flow problem, the goal is to transport a predefined amount of flow from supply to demand nodes through a given network such that the transportation costs are minimized. Usually, such networks can be described by weighted, directed graphs with capacity restrictions. In fact, the shortest path problem is a special case of the minimum cost flow problem, as source nodes can be described as supply nodes and target nodes as demand nodes. For an in-depth overview of the single-objective minimum cost flow problem, we refer to the textbook by~\citet{Ahuja1993}. For the multi-objective setting, we refer to~\citet{Hamacher2007}.

\begin{definition}[Multi-objective Minimum Cost Flow Problem]
Given a directed graph $D=(V,A)$ and a cost function $c\colon A \to \mathbb{Z}^p$ that assigns a $p$-dimensional cost vector $c(a) = ( c_1(a), \dots, c_p(a) )^\top$ to each arc $a \in A$. Let $u(a)$ denote integer-valued, non-negative, finite upper capacity bounds, and $f(a)$ an integer-valued flow, respectively, for each arc $a \in A$. Moreover, given a function $b \colon V \to \Z$ that assigns a supply/demand value $b(v)$ to each node $v \in V$, such that $\sum_{v \in V} b(v) = 0$. Let $f^{\text{\emph{\,in}}}(v)$ and $f^{\text{\emph{\,out}}}(v)$ denote the total incoming and total outgoing flow of each node $v \in V$, respectively. Then, the \emph{multi-objective integer minimum cost flow problem} is given by

\begin{equation}\label{eq:MOIMCF}\tag{MOIMCF}
    \begin{aligned}
        \min \quad & c(f) = \left( c_1(f), \dots, c_p(f) \right)^{\top} \\
        \text{s.\,t.} \quad & f^{\text{out}}(v) - f^{\text{in}}(v) = b(v) && \forall v \in V \\
        & 0 \leq f(a) \leq u(a) && \forall a \in A \\
        & f(a) \in \Z_{\geqq} && \forall a \in A
    \end{aligned}
\end{equation}
with $c_i(f) = \sum_{a \in A} f(a)\,c_i(a)$ for $i = 1, \dots, p$.  
    
\end{definition}

Nodes are classified as supply nodes if $b(v)>0$, demand nodes if $b(v)<0$, and transshipment nodes if $b(v)=0$. An illustrative example of a bi-objective integer minimum cost flow problem and the corresponding outcome space, including the visualization of different non-dominated point sets, is given in Figure~\ref{fig:biExample}.

Since the multi-objective shortest path problem is a special case of~\ref{eq:MOIMCF}, it immediately follows that~\ref{eq:MOIMCF} is intractable and its decision version is $\mathbf{NP}$-hard. It also follows that the number of extreme supported non-dominated points can grow exponentially with the number of vertices in the network. However,~\citet{raith09} show that in integer bi-objective network flow problems, the number of unsupported non-dominated points is typically much larger than that of supported or extreme supported solutions, making efficient representation methods for $\Y_N$ even more relevant. While~\citet{eusebio14} investigate representations for the bi-objective case, corresponding approaches for general multi-objective integer minimum cost flow problems remain rare. Existing work mainly considers restricted subsets of efficient solutions rather than representative sets optimized with respect to a specific quality measure~\citep{Hamacher2007}.

As noted by~\citet{Schrijver2003}, feasible solutions of network flow problems differ only along residual cycles or, particularly in the case of a so-called \emph{tree solution} (i.\,e.\ basic feasible solution), only in a linear number of induced cycles~\citep{KONEN2022333}. This affects the number of supported non-dominated points. Transitioning from one extreme efficient solution to another, where the corresponding extreme supported non-dominated points are adjacent, might be achieved by augmenting flow along a single residual cycle. While this augmentation may affect only one cycle, the corresponding residual capacity of this cycle might be large. As a result, augmenting flow incrementally along this cycle can generate a substantial number of supported non-dominated points on the face connecting these two extreme points. Additionally, other combinations of residual cycles can yield further supported non-dominated points on the same face. This contrasts with binary problems where such \textit{intermediate} integer solutions often do not exist or are sparse. Consequently, we expect the set of supported points $\Y_S$ to be significantly larger and denser than $\Y_{ES}$ in instances with high capacities, a distinction that is a central point of the analysis in Chapter~\ref{sec:RepQualSuppNonPoint}.

\definecolor{cost1}{RGB}{140, 50, 170}
\definecolor{cost2}{RGB}{0, 128, 128}

\begin{figure}[htb]
\begin{minipage}{0.48\textwidth}
	\centering
	  	\tikzstyle{vertex}=[circle,fill=white,draw=black,minimum size=20pt,inner sep=0]
  	\tikzstyle{edge} = [draw,thick,->]
  	\tikzstyle{weight} = []
  	\tikzstyle{selected vertex} = [vertex, fill=red!24]
  	\tikzstyle{selected edge} = [draw,line width=5pt,-,green!50]
  	\usetikzlibrary{arrows,automata}
  	\pgfdeclarelayer{background}
  	\pgfsetlayers{background,main}
  	\begin{tikzpicture}[->,shorten >=1pt,auto,node distance=2.8cm,
  	semithick]
  	\foreach \pos/\name in {{(0,0)/1}, {(1.5,1.5)/2}, {(1.5,-1.5)/3},
  		{(3.5,1.5)/4}, {(5,0)/5}}
  	\node[vertex] (\name) at \pos {$\name$};

    \path (1) edge[->,thick] node[sloped] {\scriptsize(10,\color{cost1}3\color{black},\color{cost2}5\color{black})} (2)
    (1) edge[->,thick] node[sloped,below] {\scriptsize(5,\color{cost1}8\color{black},\color{cost2}1\color{black})} (3)
    (2) edge[->,thick] node[sloped] {\scriptsize(4,\color{cost1}5\color{black},\color{cost2}5\color{black})} (3)
    (2) edge[->,thick] node[sloped] {\scriptsize(7,\color{cost1}3\color{black},\color{cost2}9\color{black})} (4)
    (3) edge[->,thick] node[sloped,below] {\scriptsize(8,\color{cost1}2\color{black},\color{cost2}7\color{black})} (4)
    (3) edge[->,thick] node[sloped,below] {\scriptsize(6,\color{cost1}10\color{black},\color{cost2}2\color{black})} (5)
    (4) edge[->,thick] node[sloped] {\scriptsize(8,\color{cost1}1\color{black},\color{cost2}4\color{black})} (5)
    ; 
    
    \node at (-0.8,0) {$10$};
    \node at (5.8,0) {$-10$};
    
    \node at (2.5,-2.5) {$(u_{ij}$,\color{cost1}$c^1_{ij}$\color{black},\color{cost2}$c^2_{ij}$\color{black}$)$};
\end{tikzpicture}
\end{minipage}
\hfill
\begin{minipage}{0.5\textwidth}
	\centering
	\tikzset{every picture/.style={line width=0.75pt}} 
	\tikzstyle{vertex3}=[circle,fill=white,draw=black,minimum size=2.5pt,inner sep=0]
	\tikzstyle{N_point}=[draw, circle, fill=black,minimum size=2pt,inner sep=0]
	\tikzstyle{vertex2}=[diamond,fill=blue,draw=blue,minimum size=4pt,inner sep = 0]
	\tikzstyle{vertex}=[rectangle,fill=orange,draw=orange,minimum size=3pt,inner sep = 0] 
		\newcommand{\Crossa}{$\mathbin{\tikz [x=1ex,y=1ex,line width=.1ex, black] \draw (0,0) -- (1,1) (0,1) -- (1,0);}$}%
	\begin{tikzpicture}[x=3pt,y=3pt,scale=0.7]
		\draw[->] (0,0) -- (60,0) node[right] {$c_1$};
		\draw[->] (0,0) -- (0,60) node[above] {$c_2$};
		
		\foreach \x in {10,20,30,40,50}{
		\pgfmathtruncatemacro\result{90+\x} 
		\draw (\x,0) -- (\x,-2) node[below] {\scriptsize \result};}
		
		\foreach \y in {10,20,30,40,50}
			\pgfmathtruncatemacro\result{90+\y}
		\draw (0,\y) -- (-2,\y) node[left] {\scriptsize \result};

		\draw [draw= white, fill= Azure3!30  ,fill opacity=1 ]   (60,60) -- (6,60) -- (6,54) -- (14,42) -- (35,15) -- (46,9) -- (60,9);

		\draw [fill={rgb, 255:red, 255; green, 255; blue, 255 }  ,fill opacity=0 ]   (6,54) -- (14,42) -- (35,15) -- (46,9);

            	
		\node[vertex3] at ( 13 , 54 ) {};
		\node[vertex3] at ( 10 , 57 ) {};
		\node[vertex3] at ( 17 , 48 ) {};
		\node[vertex3] at ( 24 , 39 ) {};
		\node[vertex3] at ( 14 , 51 ) {};
		\node[vertex3] at ( 21 , 42 ) {};
		\node[vertex3] at ( 28 , 33 ) {};
		\node[vertex3] at ( 35 , 24 ) {};
		\node[vertex3] at ( 18 , 45 ) {};
		\node[vertex3] at ( 25 , 36 ) {};
		\node[vertex3] at ( 32 , 27 ) {};
		\node[vertex3] at ( 39 , 18 ) {};
		\node[vertex3] at ( 17 , 57 ) {};
		\node[vertex3] at ( 24 , 48 ) {};
		\node[vertex3] at ( 21 , 51 ) {};
		\node[vertex3] at ( 28 , 42 ) {};
		\node[vertex3] at ( 35 , 33 ) {};
		\node[vertex3] at ( 18 , 54 ) {};
		\node[vertex3] at ( 25 , 45 ) {};
		\node[vertex3] at ( 32 , 36 ) {};
		\node[vertex3] at ( 39 , 27 ) {};
		\node[vertex3] at ( 46 , 18 ) {};
		\node[vertex3] at ( 22 , 48 ) {};
		\node[vertex3] at ( 29 , 39 ) {};
		\node[vertex3] at ( 36 , 30 ) {};
		\node[vertex3] at ( 43 , 21 ) {};
		\node[vertex3] at ( 50 , 12 ) {};
		\node[vertex3] at ( 24 , 57 ) {};
		\node[vertex3] at ( 28 , 51 ) {};
		\node[vertex3] at ( 35 , 42 ) {};
		\node[vertex3] at ( 25 , 54 ) {};
		\node[vertex3] at ( 32 , 45 ) {};
		\node[vertex3] at ( 39 , 36 ) {};
		\node[vertex3] at ( 46 , 27 ) {};
		\node[vertex3] at ( 22 , 57 ) {};
		\node[vertex3] at ( 29 , 48 ) {};
		\node[vertex3] at ( 36 , 39 ) {};
		\node[vertex3] at ( 43 , 30 ) {};
		\node[vertex3] at ( 50 , 21 ) {};
		\node[vertex3] at ( 26 , 51 ) {};
		\node[vertex3] at ( 33 , 42 ) {};
		\node[vertex3] at ( 40 , 33 ) {};
		\node[vertex3] at ( 47 , 24 ) {};
		\node[vertex3] at ( 54 , 15 ) {};
		\node[vertex3] at ( 35 , 51 ) {};
		\node[vertex3] at ( 32 , 54 ) {};
		\node[vertex3] at ( 39 , 45 ) {};
		\node[vertex3] at ( 46 , 36 ) {};
		\node[vertex3] at ( 29 , 57 ) {};
		\node[vertex3] at ( 36 , 48 ) {};
		\node[vertex3] at ( 43 , 39 ) {};
		\node[vertex3] at ( 50 , 30 ) {};
		\node[vertex3] at ( 33 , 51 ) {};
		\node[vertex3] at ( 40 , 42 ) {};
		\node[vertex3] at ( 47 , 33 ) {};
		\node[vertex3] at ( 54 , 24 ) {};
		\node[vertex3] at ( 30 , 54 ) {};
		\node[vertex3] at ( 37 , 45 ) {};
		\node[vertex3] at ( 44 , 36 ) {};
		\node[vertex3] at ( 51 , 27 ) {};
		\node[vertex3] at ( 58 , 18 ) {};

		\foreach \x/\y in {6/54, 14/42, 35/15,46/9}
			\node[vertex] at (\x,\y) {};
		
		\foreach \x/\y in {10/48, 21/33, 28/24}
		\node[vertex2] at (\x,\y) {};

		\foreach \x/\y in { 13/45, 17/39, 24/30}
		\node[N_point] at (\x,\y) {};
		
		\node at (47,56) {\scriptsize $\conv(\Y)+\R^2_{\geqq}$};
		
		
		
		\node[vertex][anchor= west] at (0,-10) {};
		\draw (2,-10) node [anchor= west][inner sep=1pt]   [align=left] {\scriptsize extreme supported non-dominated};
		
		\node[vertex2][anchor= west] at (0,-14) {};
		\draw (2,-14) node [anchor= west][inner sep=1pt]   [align=left] {\scriptsize supported non-dominated};
		
		\node[N_point][anchor= west] at (0,-18) {};
		\draw (2,-18) node [anchor= west][inner sep=1pt]   [align=left] {\scriptsize unsupported non-dominated};
		
		\node[vertex3] at (1,-22) {};
		\draw (2,-22) node [anchor= west][inner sep=1pt]   [align=left] {\scriptsize dominated};

	\end{tikzpicture}
\end{minipage}
	\caption{Illustration of a BOIMCF problem and its image in the objective space. The example includes ten distinct non-dominated points,  and $83$ non-efficient flows. Note that the origin of the image space is at $(90,90)^{\top}$ and not all dominated vectors are displayed.}\label{fig:biExample}
\end{figure}
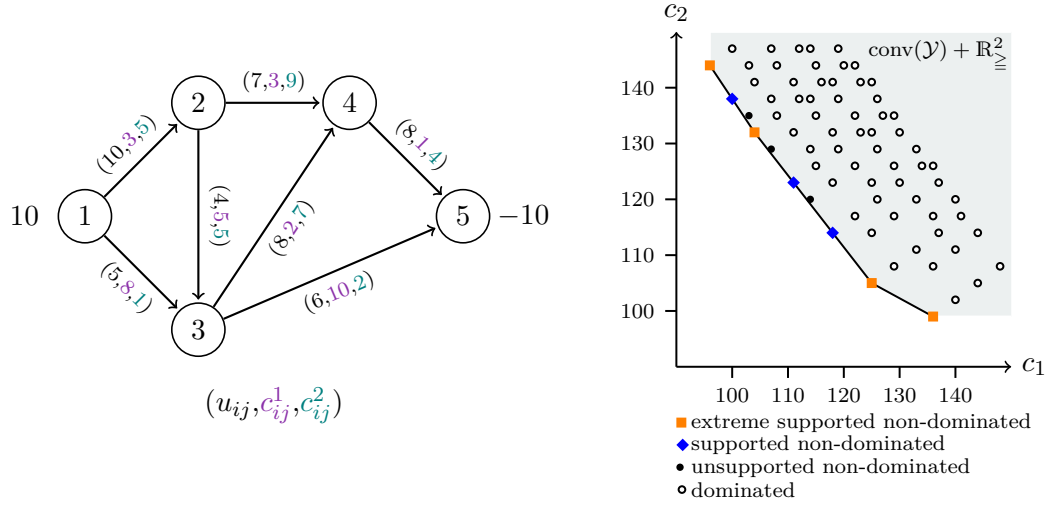

\section{Representations}\label{sec:Rep}

Representation algorithms aim to determine a representative subset $R \subseteq \Y$, referred to as representation, that is either optimal with respect to a chosen indicator or guarantees a certain quality level. While approximation algorithms usually determine a subset of feasible solutions, which may include dominated points~\citep{Bazgan2017}, we consider representations consisting of non-dominated points only. For the purpose of evaluating the representations presented in our study, we assume the complete set $\Y_N$ is known \emph{a~priori}.

\subsection{Quality Metrics and Evaluation}\label{sec:Quality}

Depending on the application, different characteristics such as cardinality, diversity, uniformity, or coverage are relevant to accurately represent the non-dominated set. In addition to cardinality, which reflects the manageability of the representative set, we introduce further quality indicators considered in this study. 

\paragraph{Coverage Error}
The coverage error, first introduced in~\citet{Sayin2000}, is defined as the maximum distance from any point $y \in \Y_N$ to its closest representation point $r \in R$. Formally, the coverage error is computed as:

$$
\operatorname{CE}(R)=\max _{y \in \Y_N} \min _{r \in R} d(y, r),
$$

where $d(\cdot, \cdot)$ denotes the distance between the two vectors. Throughout this paper, we use the weighted Tchebycheff distance with weights that are reciprocally proportional to the criteria ranges:
\begin{align*}
    &d\bigl(y^1, y^2\bigr)=\max_{j=1, \ldots, p} \left\{w_j \bigl|y_j^1-y_j^2\bigr|\right\}\\
    &\qquad \text{with} \;\,
    w_j=\frac{1}{\max _{y \in \Y_N} y_j-\min _{y \in \Y_N} y_j}, \qquad j=1, \ldots, p.
\end{align*}

As the coverage error captures worst-case distances, it indicates whether parts of the non-dominated set are missed and thus reflects how well $\Y_N$ is covered by a representative set $R$. A small value is indeed desirable, as it means that for any non-dominated point, there exists a solution in the representation whose image is sufficiently close. Note that the coverage error corresponds to the directed Hausdorff distance from $\Y_N$ to $R$, i.\,e.\ $\operatorname{CE}(R) = d_H(\Y_N,R)$. As we assume $R \subseteq \Y_N$, the directed Hausdorff distance in the opposite direction $d_H(R,\Y_N)$ equals zero.

\Cref{fig:CV-error} illustrates the coverage errors of the BOIMCF problem given in~\Cref{fig:biExample}, depicting the supported or extreme supported non-dominated points as representations.

 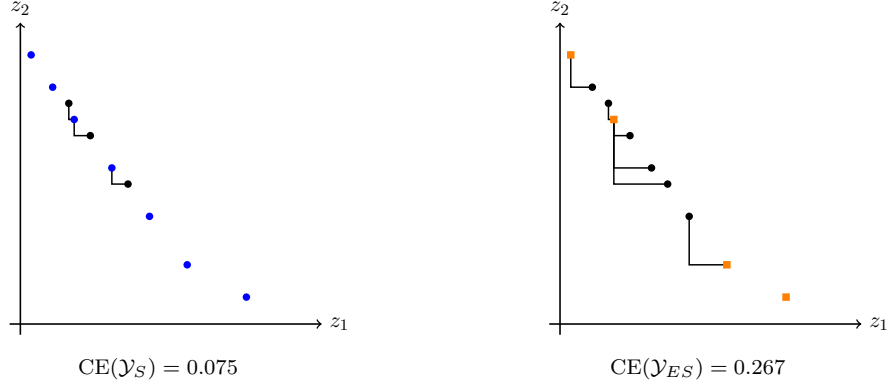
\begin{figure}
    \centering
    		\begin{minipage}[]{0.4\textwidth}
			\centering
			\scalebox{0.75}{
				\tikzset{every picture/.style={line width=0.75pt}} 
				\tikzstyle{vertex}=[circle,fill=DodgerBlue4,draw=DodgerBlue4,minimum size=3pt,inner sep=0]
				\tikzstyle{vertex2}=[circle,fill=blue,draw=blue,minimum size=3pt,inner sep = 0]
				\tikzstyle{vertex3}=[circle,fill=black,draw=black,minimum size=3pt,inner sep = 0] 
				\newcommand{\Crossa}{$\mathbin{\tikz [x=1ex,y=1ex,line width=.1ex, black] \draw (0,0) -- (1,1) (0,1) -- (1,0);}$}
				\begin{tikzpicture}[x=3pt,y=3pt,scale=0.9]
					\draw[thick,->] (2,4) -- (60,4) node[right] {$z_1$};
					\draw[thick,->] (4,2) -- (4,60) node[above] {$z_2$};

					\draw[] (14,42) -- (13,42) -- (13,45); 
					\draw[] (14,42) -- (14,39) -- (17,39);
					\draw[] (21,33) -- (21,30) -- (24,30);
					
					\foreach \x/\y in {6/54, 14/42, 35/15,46/9}
					\node[vertex2] at (\x,\y) {};
					
					\foreach \x/\y in {10/48, 21/33, 28/24}
					\node[vertex2] at (\x,\y) {};

					\foreach \x/\y in { 13/45, 17/39, 24/30}
					\node[vertex3] at (\x,\y) {};

			\end{tikzpicture}}
			{\scriptsize
			$ \operatorname{CE}(\mathcal{Y}_{S}) = 0.075 $}
			
		\end{minipage}
		\hspace{1cm}
		\begin{minipage}[]{0.4\textwidth}
			\centering
			\scalebox{0.75}{
				\tikzset{every picture/.style={line width=0.75pt}}  
				\tikzstyle{vertex}=[rectangle,fill=orange,draw=orange,minimum size=3pt,inner sep=0]
				\tikzstyle{vertex2}=[rectangle,fill=orange,draw=orange,minimum size=3pt,inner sep = 0]
				\tikzstyle{vertex3}=[circle,fill=black,draw=black,minimum size=3pt,inner sep = 0] 
				\newcommand{\Crossa}{$\mathbin{\tikz [x=1ex,y=1ex,line width=.1ex, black] \draw (0,0) -- (1,1) (0,1) -- (1,0);}$}   
				\begin{tikzpicture}[x=3pt,y=3pt,scale=0.9]
					\draw[thick,->] (2,4) -- (60,4) node[right] {$z_1$};
					\draw[thick,->] (4,2) -- (4,60) node[above] {$z_2$};

					\draw[] (6,54) -- (6, 48) -- (10, 48);
					\draw[] (14,42) -- (13,42) -- (13,45); 
					\draw[] (14,42) -- (14,39) -- (17,39);
					\draw[] (14,42) -- (14,30) -- (24,30);
					\draw[] (14,42) -- (14,33) -- (21,33);
					\draw[] (28,24) -- (28,15) -- (35,15);
					
					\foreach \x/\y in {6/54, 14/42, 35/15,46/9}
					\node[vertex] at (\x,\y) {};
					
					\foreach \x/\y in {10/48, 21/33, 28/24}
					\node[vertex3] at (\x,\y) {};

					\foreach \x/\y in { 13/45, 17/39, 24/30}
					\node[vertex3] at (\x,\y) {};

			\end{tikzpicture}}
			{\scriptsize
			$ \operatorname{CE}(\mathcal{Y}_{ES}) = 0.267 $}
		\end{minipage}
    \caption{Coverage error for $\Y_N$ in the BOIMCF problem (\Cref{fig:biExample}). Left: Supported non-dominated points (blue) used as the representation of $\Y_N$. Right: Representation by the extreme supported non-dominated points (orange).}
    \label{fig:CV-error}
\end{figure}

\paragraph{Uniformity} Uniformity aims to evaluate the dispersion of a representative subset $R$ by measuring the minimum distance between any pair of points within $R$:
$$
\operatorname{UN}(R) = \min_{r, r' \in R, r \neq r'} d(r, r')
$$

As we usually try to avoid redundancy in the representation, i.\,e., ensuring that no two points in the representation are too close to each other, a high level of uniformity is preferred. However, the level of uniformity decreases with increasing cardinality of the representation $R$. Thus, we employ this measure only in~\Cref{sec:FixedSize}, where we investigate representations of a small fixed number of alternative solutions.

\paragraph{Hypervolume Ratio} Another commonly used measure to assess the quality of a representation is the so-called \emph{hypervolume ratio}. The hypervolume ratio of a representation~$R$ for the non-dominated set \(\Y_N\) is defined as

$$\operatorname{HVR}(R)=\frac{\operatorname{HV}(R)}{\operatorname{HV}\left(\Y_N\right)},$$
The hypervolume indicator, denoted by $\operatorname{HV}(\, \cdot \,)$, is a well-established metric in evolutionary multi-objective optimization~\citep{zitzler2003performance}. It measures the volume of the region dominated by a set of representative points and is bounded by a pre-specified reference point (often the Nadir point). A larger hypervolume indicates a greater dominated region, reflecting a higher-quality approximation of the solution set. In the following, we use the hypervolume ratio to quantify the quality of the representation $R$ relative to the complete set $\Y_N$ that it represents.

Note that the selection of the reference point can significantly influence the calculated hypervolume and the hypervolume ratio. For example, in a bi-objective setting, where the Nadir point is used as the reference point, the two lexicographically optimal solutions will not contribute to the hypervolume. In our experiments, we therefore shift the Nadir point by plus one unit in each component to determine the reference point. However, the influence of the lexicographical minima of the hypervolume remains limited.

\begin{figure}
    \centering
   		\begin{minipage}[]{0.4\textwidth}
			\centering
			\scalebox{0.75}{
				\tikzset{every picture/.style={line width=0.75pt}}  
				\tikzstyle{vertex}=[circle,fill=DodgerBlue4,draw=DodgerBlue4,minimum size=3pt,inner sep=0]
				\tikzstyle{vertex2}=[circle,fill=blue,draw=blue,minimum size=3pt,inner sep = 0]
				\tikzstyle{vertex3}=[circle,fill=black,draw=black,minimum size=3pt,inner sep = 0] 
				\newcommand{\Crossa}{$\mathbin{\tikz [x=1ex,y=1ex,line width=.1ex, black] \draw (0,0) -- (1,1) (0,1) -- (1,0);}$}  
				\begin{tikzpicture}[x=3pt,y=3pt,scale=0.9]
					\draw[thick,->] (2,4) -- (60,4) node[right] {$z_1$};
					\draw[thick,->] (4,2) -- (4,60) node[above] {$z_2$};
					
					\coordinate (ref) at (47, 55);
					
					\foreach \x/\y in {13/45, 17/39, 24/30} {
						\fill[black , opacity=0.2] (\x,\y) -- (\x,55) -- (ref) -- (47,\y) -- cycle;
						\node[vertex3] at (\x,\y) {};
					}
					
					\foreach \x/\y in {6/54, 14/42, 35/15,46/9} {
						\fill[blue, opacity=0.6] (\x,\y) -- (\x,55) -- (ref) -- (47,\y) -- cycle;
						\node[vertex2] at (\x,\y) {};
					}
					
					\foreach \x/\y in {10/48, 21/33, 28/24} {
						\fill[blue, opacity=0.6] (\x,\y) -- (\x,55) -- (ref) -- (47,\y) -- cycle;
						\node[vertex2] at (\x,\y) {};
					}
					
					\node[circle, fill=gray, draw=gray, minimum size=3pt, inner sep=0] at (ref) {};
					
					\node[above right] at (ref) {$r$};
					
					\draw[dashed, gray] (47,4) -- (47,55);
					\draw[dashed, gray] (4,55) -- (47,55);
					
				\end{tikzpicture}
			}
            {\scriptsize
			$ \operatorname{HVR}(\mathcal{Y}_{S}) = 0.97 $}
		\end{minipage}
		\hspace{1cm}
		\begin{minipage}[]{0.4\textwidth}
			\centering
			\scalebox{0.75}{
				\tikzset{every picture/.style={line width=0.75pt}} 
				\tikzstyle{vertex}=[rectangle,fill=orange,draw=orange,minimum size=3pt,inner sep=0]
				\tikzstyle{vertex2}=[circle,fill=blue,draw=blue,minimum size=3pt,inner sep = 0]
				\tikzstyle{vertex3}=[circle,fill=black,draw=black,minimum size=3pt,inner sep = 0] 
				\newcommand{\Crossa}{$\mathbin{\tikz [x=1ex,y=1ex,line width=.1ex, black] \draw (0,0) -- (1,1) (0,1) -- (1,0);}$}%
				\begin{tikzpicture}[x=3pt,y=3pt,scale=0.9]
					\draw[thick,->] (2,4) -- (60,4) node[right] {$z_1$};
					\draw[thick,->] (4,2) -- (4,60) node[above] {$z_2$};
					
					\coordinate (ref) at (47, 55);

					\foreach \x/\y in {10/48, 21/33, 28/24} {
						\fill[black, opacity=0.2] (\x,\y) -- (\x,55) -- (ref) -- (47,\y) -- cycle;
						\node[vertex3] at (\x,\y) {};
					}
					
					\foreach \x/\y in {13/45, 17/39, 24/30} {
						\fill[, opacity=0.2] (\x,\y) -- (\x,55) -- (ref) -- (47,\y) -- cycle;
						\node[vertex3] at (\x,\y) {};
					}
					
					\foreach \x/\y in {6/54, 14/42, 35/15,46/9} {
						\fill[orange, opacity=0.6] (\x,\y) -- (\x,55) -- (ref) -- (47,\y) -- cycle;
						\node[vertex] at (\x,\y) {};
					}
					
					\node[circle, fill=gray, draw=gray, minimum size=3pt, inner sep=0] at (ref) {};
					
					\node[above right] at (ref) {$r$};
					
					\draw[dashed, gray] (47,4) -- (47,55);
					\draw[dashed, gray] (4,55) -- (47,55);
					
				\end{tikzpicture}
			}
			{\scriptsize
			$ \operatorname{HVR}(\mathcal{Y}_{ES}) = 0.73 $
            }
		\end{minipage}
    \caption{Hypervolume ratio for the BOIMCF non-dominated set (\Cref{fig:biExample}). Left: Supported non-dominated points (blue) used as the representation of $\Y_N$. Right: Representation by the extreme supported non-dominated points (orange).}
    \label{fig:HV-ratio}
\end{figure}
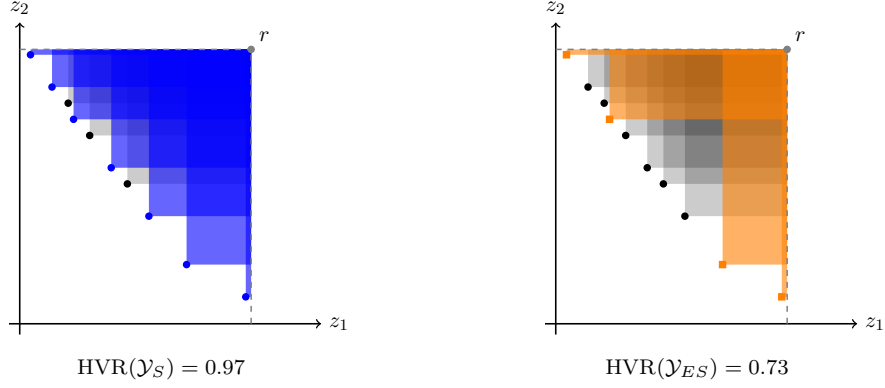

\paragraph{$\varepsilon$-indicator}
The (multiplicative) $\varepsilon$-indicator is a widely recognized performance measure for approximation algorithms~\citep{Papadimitriou}. This indicator evaluates how well a representation set 
$R$ approximates the set $\Y_N$ by determining the smallest scaling factor $\varepsilon$ required such that every point in $\Y_N$ is weakly dominated by at least one point in $R$. 
Following~\citet{zitzler2003performance} and~\citet{vaz15representation}, it is calculated as:
\begin{align*}
&I_\varepsilon(R)=\max _{y \in \Y_N} \min _{r \in R} \varepsilon(r, y),\\
&\quad\text{where}\quad
\varepsilon(r, y)=\max _{i \in\{1,\ldots,p\}}\frac{r_i}{y_i}
\end{align*}

The  $\varepsilon$-indicator provides an intuitive measure of approximation quality. If $R=\Y_N$, it holds $I_\varepsilon(R)=1$. The excess above 1.0 represents the maximum relative approximation error across all objectives. Indeed, the factor $I_{\varepsilon}(R)$ directly corresponds to the approximation ratio $(1+\varepsilon)$ that is typically used in approximation theory; see, for example,~\citet{Papadimitriou}.  This concept is closely related to the coverage representation problem with a different notion of ``distance''; achieving good coverage often corresponds to obtaining a good $\varepsilon$-indicator, and vice versa, as demonstrated in~\citet{vaz15representation}.

\section{Evaluation of Supported Points as Representation}\label{sec:RepQualSuppNonPoint}

In this section, we evaluate numerically the representation quality of the sets $\Y_S$ and $\Y_{ES}$. It is assessed using several different indicators across a diverse set of test classes and instances. Consistent with the framework established in~\Cref{sec:NetworkProb}, the experiments are organized in two primary categories: binary and low-capacity problems, and~\ref{eq:MOIMCF} problems with higher arc capacities.

The remainder of this section is structured as follows. First, we describe the experimental setup, including the generation process for the different test instances. Next, we present how to extract the supported and extreme supported non-dominated points from the set of non-dominated points, followed by a detailed analysis and discussion of the numerical results.

Note on Statistical Evaluation: While formal hypothesis testing is common in randomized heuristic benchmarks, our study investigates deterministic geometric properties across complete non-dominated sets of solved network instances. Because the performance disparities between $\mathcal{Y}_S$ and $\mathcal{Y}_{ES}$, particularly in coverage errors and hypervolume ratios, are systemic and span orders of magnitude, we report aggregated minimum, maximum, and mean values over 15 instances per configuration to robustly capture the structural scaling behavior across network topologies.

\subsection{Experimental Setup}\label{sec:4.1}
The general setup of the computational experiments for each test instance is as follows.  
First, the complete non-dominated point set is computed. For~\ref{eq:MOIMCF},~\ref{eq:MOSP}, and~\ref{eq:BCTP}/~\ref{eq:BGCTP}, the open-source implementation of the Defining Point Algorithm by~\citet{Dachert2024} is used. As~\ref{eq:MOMST} integer programming formulations typically require exponentially many subtour elimination constraints, we use the problem-specific algorithm and implementation presented in~\citet{MaristanyAlgMMST2025} and \citet{MaristanyImpleMMST2023}. The subsets of $\Y_S$ and $\Y_{ES}$ are then extracted from $\Y_N$ by solving modified versions of subproblems introduced in~\citet{Serpil2024}, as detailed in~\Cref{sec:indent}. These are solved using C++ code that utilizes the CPLEX Callable Library (version 22.1.1). 

Metric evaluations of coverage and $\varepsilon$-indicator are implemented in C++. For computations regarding the hypervolume, we use the open-source C implementation (Version 1.3) described in~\citet{Fonseca2006}. Recall that we shift the instance-specific Nadir point by one unit in each coordinate to determine the reference point.

All computations are conducted on a computer with an Intel\textregistered\ Core\texttrademark i8-8700U CPU 3.20 GHz processor with 32 GB RAM, using a Linux operating system. 
As this is primarily an empirical study focused on representation quality, we do not report runtime performance. However, we note that computing the full non-dominated set can become computationally intensive, even for moderately sized instances. More precisely, we encountered scalability limitations across all considered problem classes for instances with more than five objectives. Moreover, the computational effort highly depends on the graph size and varies across different problem classes, which is why we consider different graph sizes for all problems. For a more efficient generation of $\Y_N$ in demanding cases, parallel implementations as proposed in~\cite{Prinz2025} might offer a practical solution.

\paragraph{Instance Generation} 
Across all considered problem classes, instance generation follows standard procedures for network optimization problems adapted to multi-objective contexts. Thereby, we utilize the NETGEN network generator~\citep{Klingman1974NETGENAP}, a widely used instance generator for network optimization problems, particularly for minimum cost flow problems. For~\ref{eq:MOMST} and~\ref{eq:BCTP}/~\ref{eq:BGCTP}, graphs are generated by first constructing random spanning trees and then adding a fixed number of arcs, depending on the desired graph size, with arc selection and cost assignment based on the NETGEN random number generator. Unless stated otherwise, arc costs lie in the interval $[1,10]$ for~\ref{eq:MOMST} instances and in $[1,100]$ for~\ref{eq:BCTP}/~\ref{eq:BGCTP} instances. 

While the multi-objective shortest path problem is indeed a special case of the minimum cost flow problem, we do not use NETGEN networks to generate~\ref{eq:MOSP} instances since they yield very few efficient paths, as observed in~\cite{SkriverAndersen2000}. Instead, we construct test instances similar to the NetMaker networks proposed in~\cite{SkriverAndersen2000} and used in, e.\,g.,~\cite{RaithEhrgott2009} and~\cite{DeLasCasas2021}. Since the original NetMaker code was not available to us, we independently implemented the generator with slight modifications. An instance is generated in the following manner: First, to ensure strong connectivity, a random Hamiltonian cycle is constructed starting at the source node $1$. The position of the sink node $n$ within this cycle is sampled from a normal distribution with a mean of $2\,n/3$ and a standard deviation of $\max\{2, n/10\}$, ensuring a sufficiently long path distance between the source and the sink. To control the network density, each node is assigned a random number of additional outgoing arcs, bounded by specified minimum and maximum degrees. The target nodes for these supplementary arcs are drawn uniformly at random from a localized neighborhood of candidate nodes restricted to a given \textit{interval} along the Hamiltonian cycle. During this step, the addition of symmetric arcs (2-cycles) is explicitly forbidden. Finally, arc costs for each objective are generated by drawing uniformly and independently from the interval $[1,10]$ unless stated otherwise. Python code of the instance generation can be found in \cite{koenen2026supprepresentations}.

For~\ref{eq:MOIMCF}, we generate instances using NETGEN with adapted parameter settings, i.\,e., the number of nodes and arcs, the distribution of supply and sink nodes, the bounds on arc costs, capacities, and total network supply. Unless stated otherwise, each instance features two supply nodes, two sink nodes, a maximum arc cost of 10 across all objectives, and an initial total supply of 50. We also ensured that no instance was generated with a feasible ideal point. For all test classes considered, defined by varying instance characteristics (e.\,g., size and structure), a set of 15 instances is generated. All test instances are publicly available, see \cite{koenen2026supprepresentations}.

\subsection{Identifying Supported and Extreme Supported Points}\label{sec:indent} 
Given the complete list of non-dominated points, modified versions of the two linear programs introduced in~\citet{Serpil2024} are used to identify the sets of all supported and extreme supported non-dominated points. 

\paragraph{Identifying Supported non-dominated Points}
To identify whether a non-dominated point $y^k\in \Y_{N}$ is supported, i.\,e., $y^k\in \Y_S$, the linear program, denoted by $S_{y^k}$, is solved. Note that this check requires the complete non-dominated set $\Y_N$ as input.

\begin{equation}\label{eq:S_yk}\tag{$S_{y^k}$}
	\def\arraystretch{1.4}
	\begin{array}{rr@{\extracolsep{0.7ex}}l@{\extracolsep{0.7ex}}lll}
	    \max && \multicolumn{1}{l}{\displaystyle z^k }\\
		\text{s.t.}  &&\displaystyle \sum_{i=1}^p \lambda_i \,y_i^k \leq \sum_{i=1}^p  \lambda_i\, y_i^j  & \quad j=1, \ldots, |\Y_N|, j\neq k\\[2.6ex]
		&& \displaystyle  \sum_{i=1}^p \lambda_i=1 &\\
        &&\displaystyle z^k \leq \lambda_i & \quad i=1,\ldots,p  \\ 
		&&\lambda_i \geq 0  & \quad i=1, \ldots, p
	\end{array}
\end{equation}

This linear program contains \( p \) non-negative variables, representing the weights in a weighted sum subproblem~\eqref{eq:wsm}. The first set of constraints ensures that there is no other non-dominated point that is strictly better in the weighted sum with the weighting vector $\lambda$. The constraint $\sum_{i=1}^p \lambda_i = 1$ normalizes the weight vector. The third set of constraints bounds the auxiliary variable $z^k$ from above by each weight $\lambda_i$, meaning the objective effectively maximizes the minimum weight of the vector $\lambda$. Consequently, this formulation checks whether a weight vector exists for which \( y^k \) corresponds to the image of an optimal solution to~\eqref{eq:wsm}. In this case, $y^k$ is a supported non-dominated point if and only if $z^k > 0$. Consequently, we can easily determine if $y^k$ is weakly supported non-dominated, $z^k=0$ and $\lambda \in \Lambda^p_0$ or supported non-dominated, $z^k>0$ and $\lambda \in \Lambda^p$~\citep{koenen2025}.

\begin{theorem}
	A non-dominated point $y^k$ is supported, i.\,e., $y^k\in \Y_{S}$ if and only if $z^k>0$ in~\ref{eq:S_yk}.
\end{theorem} 

\begin{proof}\mbox{}
    \begin{itemize}
     \item[(i)] Suppose that $y^k \in \Y_S$. Then $y^k$ is the image of an efficient solution that solves $P_{\hat{\lambda}}$ for some $\hat{\lambda} \in \Lambda^p$; that is, $\hat{\lambda}^\top y^k \leq \hat{\lambda}^\top y$ for all $y \in \Y$. Hence, $\hat{\lambda}$ is a feasible solution to~\ref{eq:S_yk}. Since $\hat{\lambda}>0$, the condition $z^k >0$ must hold.
     
     \item[(ii)] Suppose that $z^k > 0$. Then there exists $\lambda \in \Lambda^p$ such that $\lambda^\top y^k \leq \lambda^\top y^j$ for $y^j \in \Y_N$. This implies that $\lambda^\top y^k \leq \lambda^\top y$ for any $y \in \Y$. In order to see this, note that if there exists $\hat{y} \in \Y$ such that $\lambda^\top \hat{y} < \lambda^\top y^k$, then $\lambda^\top y < \lambda^\top$ for any $y^j \in \Y_N$. This in turn implies that no element of $\Y_N$ can be the image of an optimal solution to problem $P_\lambda$, which is a contradiction. Therefore, we conclude that $\lambda^\top y^k \leq \lambda^\top y$ for all $y \in \Y$, meaning $y^k$ is the image of an optimal solution to problem $P_\lambda$ with $\lambda\in \Lambda^p$, and hence $y^k \in \Y_S$ follows.
     \end{itemize}
\end{proof}

\paragraph{Identifying Extreme Supported non-dominated Points}
To determine whether a supported non-dominated point $y^k\in \Y_{ES}$, the following linear program, denoted by $E_{y^k}$, is solved. This check requires the complete set of supported non-dominated points as input.

\begin{equation}\label{eq:E_yk}\tag{$E_{y^k}$}
	\def\arraystretch{1.4}
	\begin{array}{rr@{\extracolsep{0.7ex}}l@{\extracolsep{0.7ex}}lll}
		 \min && \multicolumn{1}{l}{ \displaystyle \alpha_k }\\[0.6ex]
		\text{s.t.}  &&\displaystyle \sum_{j=1}^{|\Y_S|} \alpha_j y_i^j = y_i^k   & \quad i=1, \ldots, p \\[3ex]
		&&\displaystyle  \sum_{j=1}^{|\Y_S|} \alpha_j=1 &\\
		&&\alpha_j \geq 0  & \quad j=1, \ldots, |\Y_S|
	\end{array}
\end{equation}
This formulation is a linear program with $|\Y_S|$ non-negative variables. The set of constraints ensure that the non-dominated point $y^k$ is a convex combination of all supported non-dominated points. If $y^k$ can be constructed as a proper convex combination, i.\,e.\ $y^k=\sum_{j=1}^{|\Y_S|} \alpha_j\, y^k$ with $\alpha_j < 1$ for all $j$, then it is not an extreme point.

\begin{theorem}[\citealt{Serpil2024}]
	A non-dominated point $y^k \in \Y_{N}$ is extreme supported, i.\,e., $y^k\in \Y_{ES}$ if and only if $\alpha_k=1$ in~\ref{eq:E_yk}.
\end{theorem}

\subsection{Binary and Low-Capacity Problems}

 The problems analyzed in this section are characterized by binary decision variables or low arc capacities, as introduced in~\Cref{subsec:BinProb}, including~\ref{eq:MOMST},~\ref{eq:MOSP},~ \ref{eq:BCTP}/~\ref{eq:BGCTP}, and~\ref{eq:MOIMCF} with unit or highly restricted arc capacities.

To evaluate representations across various binary and low-capacity structures, we consider the following three primary directions: network size and density, dimensionality of the objective space, and cost correlation. To ensure clarity and conciseness, the following tables report the mean values across the 15 randomly generated instances for selected representative configurations of each problem type.

\paragraph{Impact of Network Size and Density}
 In the first test class the impact of scaling the network topology is investigated. We examine the effects of varying the number of nodes ($n$) and arcs ($m$) simultaneously, as well as increasing network density by adding edges to a network, while keeping the number of nodes fixed. For~\ref{eq:MOMST} and~\ref{eq:MOSP}, the number of objectives is fixed at $p=3$.

\begin{table}[h!]\tiny
    \centering
    \caption{Mean numerical results for selected configurations across different problem classes. Each configuration averages 15 instances.}
    \vspace{0.4em}
    \label{tab:combined_mean_results}
        \begin{tabular}{ll r c c c c c c c c}
            \toprule
            & & & \multicolumn{2}{c}{Size {(Ratio)}} & \multicolumn{2}{c}{HVR $\uparrow$}  & \multicolumn{2}{c}{CE $\downarrow$} & \multicolumn{2}{c}{$I_\varepsilon$ $\downarrow$} \\
            \cmidrule(lr){4-5} \cmidrule(lr){6-7} \cmidrule(lr){8-9} \cmidrule(lr){10-11}
            Problem & & \multicolumn{1}{c}{$|\Y_N|$} & $\Y_S$ & $\Y_{ES}$ & $\Y_S$ & $\Y_{ES}$ & $\Y_S$ & $\Y_{ES}$ & $\Y_S$ & $\Y_{ES}$ \\
            \midrule
            \multicolumn{11}{l}{{MOMST}} \\
            $n=30$ & $m=40$  & 166.800 & 0.359 & 0.342 & 0.968 & 0.964 & 0.138 & 0.140 & 1.0207 & 1.0215 \\
                   & $m=70$  & 1147.133 & 0.226 & 0.194 & 0.978 & 0.975 & 0.083 & 0.085 & 1.0307 & 1.0342 \\
                   & $m=100$ & 2122.933 & 0.208 & 0.166 & 0.983 & 0.979 & 0.069 & 0.071 & 1.0375 & 1.0449 \\
            $n=50$ & $m=60$  & 184.200 & 0.369 & 0.337 & 0.972 & 0.968 & 0.133 & 0.135 & 1.0123 & 1.0125 \\
                   & $m=70$  & 736.933 & 0.246 & 0.212 & 0.977 & 0.973 & 0.105 & 0.107 & 1.0145 & 1.0154 \\
                   & $m=80$  & 1093.091 & 0.222 & 0.188 & 0.978 & 0.975 & 0.087 & 0.088 & 1.0172 & 1.0189 \\
            \midrule
            \multicolumn{11}{l}{{MOSP}} \\
            $n=250$ & $m=500$  & 509.067  & 0.293 & 0.286 & 0.985 & 0.985 & 0.116 & 0.116 & 1.008 & 1.008 \\
                    & $m=1000$ & 129.067  & 0.399 & 0.395 & 0.976 & 0.976 & 0.232 & 0.232 & 1.038 & 1.038 \\
            $n=1000$& $m=3000$ & 4995.067 & 0.176 & 0.170 & 0.985 & 0.984 & 0.098 & 0.098 & 1.014 & 1.014 \\
                    & $m=4000$ & 8601.800 & 0.162 & 0.152 & 0.988 & 0.987 & 0.113 & 0.113 & 1.019 & 1.019 \\
                    & $m=5000$ & 497.600  & 0.236 & 0.234 & 0.977 & 0.977 & 0.195 & 0.195 & 1.061 & 1.061 \\
                    & $m=10000$& 127.200  & 0.286 & 0.283 & 0.968 & 0.967 & 0.271 & 0.271 & 1.154 & 1.154 \\
            \midrule
            \multicolumn{11}{l}{{BCTP}} \\
            $n=100$ & $m=200$ & 180.467 & 0.116 & 0.115 & 0.988 & 0.988 & 0.155 & 0.155 & 1.0138 & 1.0138 \\
                    & $m=400$ & 167.467 & 0.129 & 0.119 & 0.990 & 0.989 & 0.206 & 0.206 & 1.0127 & 1.0130 \\
            \midrule
            \multicolumn{11}{l}{{BGCTP}} \\
            $n=50$  & $m=100$ & 183.000 & 0.133 & 0.132 & 0.977 & 0.977 & 0.147 & 0.147 & 1.0414 & 1.0414 \\
                    & $m=200$ & 138.933 & 0.127 & 0.124 & 0.937 & 0.936 & 0.141 & 0.141 & 1.0264 & 1.0264 \\
            \bottomrule
        \end{tabular}
\end{table}

Across all problem classes, when comparing graphs with the same density, increasing the number of nodes generally leads to larger non-dominated sets, while the relative proportions of supported and extreme supported points decrease. Notably, for~\ref{eq:MOSP} and~\ref{eq:BCTP}, varying the network density leads to a different effect. For a fixed number of nodes, the total number of non-dominated points decreases as the density increases. This behavior can be explained by a problem-specific characteristic. In sparse graphs, paths between source and target nodes often contain many edges, causing costs to accumulate in shortest path objectives and resulting in stronger conflicts between objectives. As the density increases and additional arcs are added, more direct paths become available. These paths effectively act as shortcuts, reducing the accumulated costs and thereby the conflicts between the objectives, leading to smaller non-dominated sets, as discussed in more detail in~\cite{Loehken2025}.

Despite this scaling behavior for~\ref{eq:MOSP} and~\ref{eq:BCTP}, all problem classes share a common characteristic: the cardinalities of $\Y_S$ and $\Y_{ES}$ remain nearly identical. Because non-extreme supported points are almost non-existent in these structures, the representation quality metrics for both sets are indistinguishable. Both sets maintain excellent representation quality across all configurations, consistently achieving a hypervolume ratio (HVR) above 0.96, low coverage errors (CE), and near-optimal $I_\epsilon$ indicators. Ultimately, the data confirm that computing only the extreme supported points is sufficient to guarantee a high-quality representation.

\paragraph{Scaling the Number of Objectives}
In the second test class, the network topology is held constant while the dimensionality of the objective space varies from $p=3$ to $p=5$. 

For~\ref{eq:MOSP} instances we chose a maximum degree of $20$, and an interval size defined as four times the maximum degree.

\begin{table}[h!]\tiny
    \centering
    \caption{Mean numerical results for increasing the number of objectives $p$. Each configuration averages 15 instances.}
    \label{tab:binary_objectives}
    \vspace{0.5em}
    \begin{tabular}{ll r c c c c c c c c}
        \toprule
        & & & \multicolumn{2}{c}{Size {(Ratio)}} & \multicolumn{2}{c}{HVR $\uparrow$}  & \multicolumn{2}{c}{CE $\downarrow$} & \multicolumn{2}{c}{$I_\varepsilon$ $\downarrow$} \\
        \cmidrule(lr){4-5} \cmidrule(lr){6-7} \cmidrule(lr){8-9} \cmidrule(lr){10-11}
        Problem & & \multicolumn{1}{c}{$|\Y_N|$} & $\Y_S$ & $\Y_{ES}$ & $\Y_S$ & $\Y_{ES}$ & $\Y_S$ & $\Y_{ES}$ & $\Y_S$ & $\Y_{ES}$ \\
        \midrule
      
      MOMST & $p=3$ &   323.5 & 0.293 & 0.271 & 0.971 & 0.968 & 0.141 & 0.141 & 1.0386 & 1.0407 \\
        $n=20, m=40$    & $p=4$ &  2883.5 & 0.171 & 0.169 & 0.973 & 0.973 & 0.155 & 0.155 & 1.0416 & 1.0416 \\
        & $p=5$ & 19528.0 & 0.108 & 0.108 & 0.976 & 0.976 & 0.170 & 0.170 & 1.0429 & 1.0429 \\
        \midrule
        MOSP & $p=3$ &  40.1 & 0.474 & 0.472 & 0.969 & 0.969 & 0.299 & 0.301 & 1.147 & 1.147 \\
               $n=20, m=40$          & $p=4$ & 106.1 & 0.379 & 0.379 & 0.980 & 0.980 & 0.331 & 0.331 & 1.163 & 1.163 \\
                        & $p=5$ & 220.9 & 0.314 & 0.314 & 0.984 & 0.984 & 0.371 & 0.371 & 1.171 & 1.171 \\
        \bottomrule
    \end{tabular}
\end{table}

As shown in~\Cref{tab:binary_objectives}, increasing the number of objectives $p$ causes the total number of non-dominated points to grow exponentially. In~\ref{eq:MOMST} instances, for example, $|\mathcal{Y}_{N}|$ scales from a mean of $323.5$ at $p=3$ to nearly $20,000$ at $p=5$, leading to a steady decrease in the relative proportions of supported and extreme supported points. 

The data confirm our hypothesis for binary and low-capacity structures: the cardinality of $\mathcal{Y}_{S}$ and $\mathcal{Y}_{ES}$ remain nearly identical across all tested dimensions, achieving similarly high values across all considered quality indicators. Consequently, even in higher dimensions, the set of extreme supported non-dominated points $\mathcal{Y}_{ES}$ is already a high-quality representation.

\paragraph{Impact of Cost Correlation}
To investigate the impact of correlations among objective function coefficients, we conducted analogous experiments using test instances similar to those introduced by~\citet{Serpil2024}. The results show that an increasing cost correlation leads to a substantial reduction of the total number of non-dominated points, while the relative proportion of (extreme) supported points increases, thereby maintaining consistently high representation quality. Detailed descriptions of these experiments, along with the complete numerical results, are available in the associated repository~\citep{koenen2026supprepresentations}.

\paragraph{Minimum Cost Flow Problems with Unit or Low Arc Capacities}\label{para:Min_Cost_Flow}

The experiments for networks with unit or highly restricted arc capacities closely mirror the behavior observed in the binary capacity problems mentioned above. To avoid redundancy and maintain clarity, we omit the results of these tests here. Readers interested in small-capacity networks are referred to~\cite{koenenphd} for comprehensive results or can find additional instances with results in the associated repository~\citep{koenen2026supprepresentations}.

\subsection{Capacitated Flow Structures}

This section addresses~\ref{eq:MOIMCF} problems where arc capacities $u(a) > 1$ allow successive flow increments along cycles, which can generate numerous integer flows between adjacent basic feasible flows and thus \emph{intermediate} supported non-dominated points on the faces of the convex hull between consecutive extreme non-dominated points. The presence of these intermediate supported non-dominated points significantly alters the structure of ``good'' representations.

In addition to the tests in the previous section, we examine two further directions: scaling supplies and capacities and benchmark instances with a high number of supply and demand nodes.

\paragraph{Scaling Supply and Capacity}

To analyze how scaling node supplies and arc capacities affects the representation quality of $\Y_S$ and $\Y_{ES}$, we conduct a scaling test based on a parameter $\mu$ as follows. Starting with a baseline instance class ($n=50$, $m=100$) where $\mu=1$, we systematically scale both the node flow balances and the upper arc capacities by factors of $\mu \in \{1,2,3,4\}$. This approach increases the overall feasible flow volume without altering the underlying graph topology, leading to the following parameterized problem formulation: \begin{equation}\label{eq:mu-MOIMCF}\tag{$\mu$-MOIMCF}
    \begin{aligned}
        \min \quad & c(f) = \left( c_1(f), \dots, c_p(f) \right)^{\top} \\
        \text{s.t.} \quad & f^{\text{out}}(v) - f^{\text{in}}(v) = \mu \cdot b(v) && \forall v \in V \\
        & 0 \leq f(a) \leq \mu \cdot u(a) && \forall a \in A \\
       & f(a) \in \Z_{\geqq} && \forall a \in A.
    \end{aligned}
    \end{equation}

\begin{table}[h!]\tiny
    \centering
    \caption{Scaling Tests for $p=3$.}\label{tab_factor}
    \vspace{0.4em}
    \begin{tabular}{ll r c c c c c c c c}
        \toprule
        & & & \multicolumn{2}{c}{Size {(Ratio)}} & \multicolumn{2}{c}{HVR $\uparrow$}  & \multicolumn{2}{c}{CE $\downarrow$} & \multicolumn{2}{c}{$I_\varepsilon$ $\downarrow$} \\
        \cmidrule(lr){4-5} \cmidrule(lr){6-7} \cmidrule(lr){8-9} \cmidrule(lr){10-11}
        Class & & \multicolumn{1}{c}{$|\Y_N|$} & $\Y_S$ & $\Y_{ES}$ & $\Y_S$ & $\Y_{ES}$ & $\Y_S$ & $\Y_{ES}$ & $\Y_S$ & $\Y_{ES}$ \\
        \midrule 
        $\mu= 1$ & min  &    2.0 & 0.242 & 0.018 & 0.918 & 0.065 & 0.000 & 0.000 & 1.000 & 1.000 \\
        & max  & 2037.0 & 1.000 & 1.000 & 1.000 & 1.000 & 0.400 & 0.800 & 1.006 & 1.040 \\
        & mean &  489.4 & 0.612 & 0.138 & 0.990 & 0.701 & 0.094 & 0.323 & 1.003 & 1.026 \\
        \midrule
        $\mu = 2$ & min  &    3.0 & 0.200 & 0.004 & 0.948 & 0.028 & 0.000 & 0.188 & 1.000 & 1.002 \\
        & max  & 8334.0 & 1.000 & 0.667 & 1.000 & 0.895 & 0.391 & 0.778 & 1.003 & 1.041 \\
        & mean & 2032.2 & 0.590 & 0.072 & 0.994 & 0.646 & 0.066 & 0.360 & 1.001 & 1.027 \\
        \midrule
        $\mu = 3$ & min  &     4.0 & 0.169 & 0.002 & 0.961 & 0.017 & 0.000 & 0.188 & 1.000 & 1.002 \\
        & max  & 18885.0 & 1.000 & 0.500 & 1.000 & 0.893 & 0.260 & 0.692 & 1.002 & 1.041 \\
        & mean &  4810.0 & 0.584 & 0.049 & 0.996 & 0.628 & 0.046 & 0.343 & 1.001 & 1.027 \\
        \midrule
        $\mu = 4$ & min  &     5.0 & 0.156 & 0.001 & 0.969 & 0.013 & 0.000 & 0.188 & 1.000 & 1.002 \\
        & max  & 33689.0 & 1.000 & 0.400 & 1.000 & 0.892 & 0.195 & 0.647 & 1.001 & 1.041 \\
        & mean &  8733.8 & 0.585 & 0.038 & 0.997 & 0.620 & 0.034 & 0.354 & 1.0010 & 1.0270 \\
        \midrule
        $\boldsymbol{\varnothing}$ & \textbf{mean} & 4016.350 & 0.593 & 0.074 & 0.994 & 0.649 & 0.060 & 0.345 & 1.002 & 1.027 \\
        \bottomrule
    \end{tabular}
\end{table}

\begin{figure}[ht]
    \centering
    \begin{minipage}[t]{0.48\textwidth}
        \centering
        \begin{tikzpicture}
            \begin{axis}[
                width=\linewidth,
                height=6cm,
                ymin=0.0, ymax=1.1, 
                xtick={1, 2, 3, 4},
                xticklabels={$\mu=1$, $\mu=2$, $\mu=3$, $\mu=4$},
                grid=major,
                legend style={
                    at={(0.5, -0.2)},
                    anchor=north,
                    legend columns=2,
                    font=\tiny,
                    draw=none,
                    /tikz/every even column/.append style={column sep=0.3cm}
                }
            ]

            \addplot[color=blue, thick, dashed, mark=*] coordinates {
                (1, 0.612) (2, 0.590) (3, 0.584) (4, 0.585)
            };
            \addlegendentry{Ratio $\frac{|\Y_S|}{|\Y_N|}$}

            \addplot[color=blue, thick, mark=*] coordinates {
                (1, 0.990) (2, 0.994) (3, 0.996) (4, 0.997)
            };
            \addlegendentry{HVR($\Y_S$)}

            \addplot[color=orange, thick, dashed, mark=square*] coordinates {
                (1, 0.138) (2, 0.072) (3, 0.049) (4, 0.038)
            };
            \addlegendentry{Ratio $\frac{|\Y_{ES}|}{|\Y_N|}$}

            \addplot[color=orange, thick, mark=square*] coordinates {
                (1, 0.701) (2, 0.646) (3, 0.628) (4, 0.620)
            };
            \addlegendentry{HVR($\Y_{ES}$)}

            \end{axis}
        \end{tikzpicture}
        \caption{Ratios \& HVR means for~\Cref{tab_factor}.}
        \label{fig:ratio_hvr_plot}
    \end{minipage}
    \hfill
    \begin{minipage}[t]{0.48\textwidth}
        \centering
        \begin{tikzpicture}
            \begin{axis}[
                width=\linewidth,
                height=6cm,
                ymin=0.0, ymax=0.9, 
                xtick={1, 2, 3, 4},
                xticklabels={$\mu=1$, $\mu=2$, $\mu=3$, $\mu=4$},
                grid=major,
                legend style={
                    at={(0.5, -0.2)},
                    anchor=north,
                    legend columns=2,
                    font=\tiny,
                    draw=none,
                    /tikz/every even column/.append style={column sep=0.3cm}
                }
            ]

            \addplot[name path=CE_Smin, draw=none, forget plot] coordinates {
                (1, 0.0) (2, 0.0) (3, 0.0) (4, 0.0)
            };
            \addplot[name path=CE_Smax, draw=none, forget plot] coordinates {
                (1, 0.4) (2, 0.391) (3, 0.26) (4, 0.195)
            };
            \addplot[blue!20, opacity=0.5, area legend] fill between[of=CE_Smin and CE_Smax];
            \addlegendentry{Range $\operatorname{CE}(\Y_S)$}

            \addplot[color=blue, thick, mark=*] coordinates {
                (1, 0.094) (2, 0.066) (3, 0.046) (4, 0.034)
            };
            \addlegendentry{Mean $\operatorname{CE}(\Y_S)$}

            \addplot[name path=CE_ESmin, draw=none, forget plot] coordinates {
                (1, 0.0) (2, 0.188) (3, 0.188) (4, 0.188)
            };
            \addplot[name path=CE_ESmax, draw=none, forget plot] coordinates {
                (1, 0.8) (2, 0.778) (3, 0.692) (4, 0.647)
            };
            \addplot[orange!30, opacity=0.5, area legend] fill between[of=CE_ESmin and CE_ESmax];
            \addlegendentry{Range $\operatorname{CE}(\Y_{ES})$}

            \addplot[color=orange, thick, mark=square*] coordinates {
                (1, 0.323) (2, 0.360) (3, 0.343) (4, 0.354)
            };
            \addlegendentry{Mean $\operatorname{CE}(\Y_{ES})$}

            \end{axis}
        \end{tikzpicture}
        \caption{Coverage Error (CE) ranges and means for~\Cref{tab_factor}.}
        \label{fig:ce_plot}
    \end{minipage}
\end{figure}

As expected, the number of non-dominated points grows significantly with the increasing scaling factor $\mu$. However, the relative behavior differs between the sets $\Y_N$, $\Y_S$, and $\Y_{ES}$. As the mean cardinality of $\Y_N$ increases from $489.4$ (base case) to $8733.8$ (factor~4), the ratio $\lvert \Y_S \rvert / \lvert \Y_N \rvert$ remains remarkably stable at approximately~0.60. This indicates that $|\Y_S|$ scales proportionally with $|\Y_N|$ under increasing supplies and capacities.  In contrast,  the same does not hold for the extreme supported set $\Y_{ES}$. In fact, the absolute number of extreme supported non-dominated points remains constant when $\mu$ is varied. Consequently, as the total number of points grows, the ratio of $ |\Y_{ES} | / | \Y_{N} |$ decreases monotonously from $0.138$ in the base case to $0.038$ for $\mu=4$, indicating that the number of extreme supported points becomes comparatively small as the feasible flow range widens. As discussed in previous sections, this is a direct consequence of the resulting network structure: larger capacities allow for larger incremental flow augmentations along residual cycles, which generate additional intermediate supported points between consecutive extreme supported points.

Investigating the representation quality of $\Y_S$ and $\Y_{ES}$, the hypervolume results (visualized in Figure~\ref{fig:ratio_hvr_plot}) show that $\operatorname{HVR}(\Y_S)$ remains consistently high across all scaling factors, maintaining mean values above 0.99. In contrast, the hypervolume ratio of $\Y_{ES}$ decreases to values between 70\% and 62\%, confirming that the extreme supported solutions alone become insufficient to represent $\Y_N$ with comparable quality. This behavior is also reflected in the coverage error and $\varepsilon$-indicator. As illustrated in Figure~\ref{fig:ce_plot}, the mean coverage error for $\Y_S$ actually improves from 0.094 to 0.034 as $\mu$ increases, with significantly shrinking variance. Conversely, $\Y_{ES}$ requires a high, volatile coverage radius, averaging around 0.35. Consequently, relying solely on extreme supported solutions as representation results in a coverage error nearly six times higher compared to a representation by the entire set of supported non-dominated points. This difference is also visible in the $\varepsilon$-indicator,  where $\Y_S$ averages a near-optimal value of 1.002. When considering only the excess error (above the ideal value of 1.0), the average error of $\Y_S$ (0.002) corresponds to merely 7.4\% of the error produced by $\Y_{ES}$ (0.027), further emphasizing the superior representation quality of the full set of  supported non-dominated points.

Overall, these results highlight that while scaling capacities and supply/demand increases the size of the non-dominated set, the relative structure of $\Y_N$ with respect to $\Y_S$ remains stable across all scaling factors, and the supported set continues to offer a high-quality representation of $\Y_N$. In contrast, as the feasible flow range grows, the representation quality of $\Y_{ES}$ decreases, with $\Y_S$ consistently yielding a higher representation quality.

\paragraph{Impact of Network Size and Density} 
Several key findings regarding variations in network size and density, i.\,e., the number of nodes $n$ and arcs $m$, can be identified from the results of the generated NETGEN instance classes summarized in \Cref{tabel2}.

\begin{table}[h!]\tiny
    \centering
    \caption{Numerical results for the different instance classes generated with NETGEN.}\label{tabel2}
    \vspace{0.4em}
    \begin{tabular}{ll r c c c c c c c c}
        \toprule
        & & & \multicolumn{2}{c}{Size {(Ratio)}} & \multicolumn{2}{c}{HVR $\uparrow$}  & \multicolumn{2}{c}{CE $\downarrow$} & \multicolumn{2}{c}{$I_\varepsilon$ $\downarrow$} \\
        \cmidrule(lr){4-5} \cmidrule(lr){6-7} \cmidrule(lr){8-9} \cmidrule(lr){10-11}
        Class & & \multicolumn{1}{c}{$|\Y_N|$} & $\Y_S$ & $\Y_{ES}$ & $\Y_S$ & $\Y_{ES}$ & $\Y_S$ & $\Y_{ES}$ & $\Y_S$ & $\Y_{ES}$ \\
        \midrule
        
        1     & min  & 3.0  & 0.234  & 0.006   & 0.972 & 0.213 & 0.0 & 0.289 & 1.0 & 1.009 \\
       $n=20,m=40$ & max  & 3269.0  & 1.0  & 0.667  & 1.0 & 0.881 & 0.187 & 0.5 & 1.005 & 1.111 \\
        & mean   & 361.0  & 0.803  & 0.156  & 0.997 & 0.645 & 0.059 & 0.396 & 1.002 & 1.046 \\
        \midrule
        
        2         & min  & 38.0  & 0.189  & 0.006   & 0.986 & 0.591 & 0.0 & 0.161 & 1.0 & 1.021 \\
        $n=20, m=80$& max  & 6952.0  & 1.0  & 0.158  & 1.0 & 0.943 & 0.222 & 0.464 & 1.009 & 1.085 \\
        & mean   & 1995.067  & 0.472  & 0.035  & 0.996 & 0.849 & 0.054 & 0.291 & 1.005 & 1.055 \\
        \midrule
        
        3       & min  & 3.0  & 0.226  & 0.02   & 0.977 & 0.634 & 0.0 & 0.15 & 1.0 & 1.002 \\
        $n=100,m=200$& max  & 2710.0  & 1.0  & 0.667  & 1.0 & 0.942 & 0.36 & 0.5 & 1.005 & 1.081 \\
        & mean   & 600.333  & 0.538  & 0.101  & 0.994 & 0.843 & 0.102 & 0.293 & 1.003 & 1.026 \\
        \midrule
        
        4          & min  & 150.0  & 0.152  & 0.005   & 0.989 & 0.786 & 0.027 & 0.12 & 1.003 & 1.011 \\
        $n=100,m=400$& max  & 15317.0  & 0.886  & 0.2  & 0.999 & 0.971 & 0.164 & 0.423 & 1.01 & 1.113 \\
        & mean   & 4123.933  & 0.363  & 0.034  & 0.996 & 0.911 & 0.079 & 0.245 & 1.006 & 1.056 \\
        \midrule
        
        5         & min  & 422.0  & 0.102  & 0.01   & 0.984 & 0.794 & 0.042 & 0.098 & 1.002 & 1.006 \\
        $n=1000,m=2000$& max  & 4252.0  & 0.391  & 0.081  & 0.998 & 0.983 & 0.173 & 0.418 & 1.004 & 1.038 \\
        & mean   & 1912.933  & 0.246  & 0.035  & 0.994 & 0.925 & 0.101 & 0.204 & 1.003 & 1.018 \\
        \midrule
        6            & min  & 176.0  & 0.109  & 0.002   & 0.993 & 0.76 & 0.019 & 0.077 & 1.004 & 1.022 \\
        $n=1000,m=4000$ & max  & 25648.0  & 0.591  & 0.119  & 0.998 & 0.98 & 0.133 & 0.388 & 1.012 & 1.249 \\
        & mean   & 8484.333  & 0.2  & 0.023  & 0.996 & 0.939 & 0.066 & 0.187 & 1.006 & 1.051 \\
 
        \midrule
        
        7         & min  & 49.0  & 0.078  & 0.007   & 0.985 & 0.397 & 0.0 & 0.105 & 1.0 & 1.006 \\
        $n=2000,m=4000$& max  & 10348.0  & 1.0  & 0.077  & 1.0 & 0.976 & 0.208 & 0.483 & 1.008 & 1.316 \\
        & mean   & 2646.8  & 0.391  & 0.039  & 0.995 & 0.867 & 0.068 & 0.257 & 1.003 & 1.06 \\
        \midrule
        
     $\boldsymbol{\varnothing}$ & \textbf{mean} & 2874.914 & 0.430 & 0.060 & 0.995 & 0.854 & 0.076 & 0.268 & 1.004 & 1.045 \\
        \bottomrule
    \end{tabular}
\end{table}

For instances with a high number of arcs, the supported set $\Y_S$ is substantially larger than the extreme supported set $\Y_{ES}$. On average across all classes, $\Y_S$ constitutes $43.0\%$ of the non-dominated set $\Y_N$, whereas $\Y_{ES}$ makes up only $6.0\%$. This cardinality advantage translates directly into superior representation quality. Across all instances, $\Y_S$ achieves an overall mean hypervolume ratio ($\operatorname{HVR}$) of $0.995$, compared to just $0.854$ for $\Y_{ES}$. Furthermore, the average coverage error ($\operatorname{CE}$) of $\Y_S$ ($0.076$) is $28.4\%$ of the error produced by $\Y_{ES}$ ($0.268$). This quality gap is also confirmed by the $\varepsilon$-indicator; the average excess error above $1.0$ for $\Y_S$ ($0.004$) is less than one-tenth of the excess error for $\Y_{ES}$ ($0.045$). Notably, the density of the network has a significant impact on the cardinality of $\Y_N$ and $\Y_S$. In classes where the arc-to-node ratio is high (e.\,g., classes 2, 4, 6), the absolute number of non-dominated points explodes, and the relative proportion of supported solutions ($|\Y_S|/|\Y_N|$) decreases compared to sparser networks.  However, the set of supported non-dominated points remains a high-quality representation with respect to the other considered quality measures.

\paragraph{Scaling the Number of Objectives}
\Cref{tab_dimensions} summarizes the results when varying the number of objectives from $p=3$ to $p=5$. Note that we use smaller network sizes ($n=10, m=20$) for these experiments due to computational limits.

\begin{table}[htb]\tiny
    \centering
    \caption{Numerical results for $(n=10,m=20)$ with different number of objectives. The classes again contain 15 instances.}\label{tab_dimensions}
    \vspace{0.4em}
    \begin{tabular}{ll r c c c c c c c c}
        & & & \multicolumn{2}{c}{Size {(Ratio)}} & \multicolumn{2}{c}{HVR $\uparrow$}  & \multicolumn{2}{c}{CE $\downarrow$} & \multicolumn{2}{c}{$I_\varepsilon$ $\downarrow$} \\
        \cmidrule(lr){4-5} \cmidrule(lr){6-7} \cmidrule(lr){8-9} \cmidrule(lr){10-11}
        & & \multicolumn{1}{c}{$|\Y_N|$} & $\Y_S$ & $\Y_{ES}$ & $\Y_S$ & $\Y_{ES}$ & $\Y_S$ & $\Y_{ES}$ & $\Y_S$ & $\Y_{ES}$ \\
        \midrule 
        
        $p=3$ & min  & 2.0  & 0.379  & 0.016   & 0.945 & 0.015 & 0.0 & 0.0 & 1.0 & 1.0 \\
        & max  & 1023.0  & 1.0  & 1.0  & 1.0 & 1.0 & 0.133 & 0.5 & 1.006 & 1.101 \\
        & mean   & 126.867  & 0.833  & 0.2  & 0.994 & 0.575 & 0.041 & 0.371 & 1.002 & 1.051 \\

        $p=4$   & min  & 2.0  & 0.452  & 0.008   & 0.991 & 0.002 & 0.0 & 0.0 & 1.0 & 1.0 \\
        & max  & 2743.0  & 1.0  & 1.0  & 1.0 & 1.0 & 0.316 & 0.519 & 1.006 & 1.113 \\
        & mean   & 296.267  & 0.82  & 0.169  & 0.998 & 0.689 & 0.071 & 0.389 & 1.002 & 1.049 \\
        
        $p=5$         & min  & 14.0  & 0.392  & 0.003   & 0.993 & 0.002 & 0.0 & 0.234 & 1.0 & 1.012 \\
        & max  & 15364.0  & 1.0  & 0.357  & 1.0 & 0.831 & 0.22 & 0.5 & 1.006 & 1.113 \\
        & mean   & 1211.467  & 0.779  & 0.098  & 0.998 & 0.639 & 0.08 & 0.427 & 1.003 & 1.06 \\
        \midrule
        
        $\boldsymbol{\varnothing}$ & \textbf{mean} & 544.867 & 0.829 & 0.159 & 0.997 & 0.636 & 0.056 & 0.393 & 1.002 & 1.053 \\
        \bottomrule
    \end{tabular}
\end{table}

\begin{figure}[ht]
    \centering
    \begin{minipage}[t]{0.48\textwidth}
        \centering
        \begin{tikzpicture}
            \begin{axis}[
                width=\linewidth,
                height=6cm,
                xlabel={},
                ylabel={},
                ymin=0.0, ymax=1.1, 
                hide obscured x ticks=false,
                xtick={1, 2, 3},
                xticklabels={$p=3$, $p=4$, $p=5$},
                x tick label style={font=\footnotesize},
                grid=major,
                legend style={
                    at={(0.5, -0.2)},
                    anchor=north,
                    legend columns=2,
                    font=\tiny,
                    draw=none
                },
                title={}
            ]

            \addplot[color=blue, thick, dashed, mark=*] coordinates {
                (1, 0.833) (2, 0.82) (3, 0.779)
            };
            \addlegendentry{Ratio $\frac{|\Y_S|}{|\Y_N|}$}

            \addplot[color=blue, thick, mark=*] coordinates {
                (1, 0.994) (2, 0.998) (3, 0.998)
            };
            \addlegendentry{HVR($\Y_S$)}

            \addplot[color=orange, thick, dashed, mark=square*] coordinates {
                (1, 0.200) (2, 0.169) (3, 0.098)
            };
            \addlegendentry{Ratio $\frac{|\Y_{ES}|}{|\Y_N|}$}

            \addplot[color=orange, thick, mark=square*] coordinates {
                (1, 0.575) (2, 0.689) (3, 0.639)
            };
            \addlegendentry{HVR($\Y_{ES}$)}

            \end{axis}
        \end{tikzpicture}
        \caption{Ratios \& HVR for~\Cref{tab_dimensions}.}
        \label{fig:ratio_hvr_dimensions}
    \end{minipage}
    \hfill
    \begin{minipage}[t]{0.48\textwidth}
        \centering
        \begin{tikzpicture}
            \begin{axis}[
                width=\linewidth,
                height=6cm,
                xlabel={},
                ylabel={},
                ymin=0.0, ymax=0.6,
                hide obscured x ticks=false,
                xtick={1, 2, 3},
                xticklabels={$p=3$, $p=4$, $p=5$},
                x tick label style={font=\footnotesize},
                grid=major,
                legend style={
                    at={(0.5, -0.25)},
                    anchor=north,
                    legend columns=2,
                    font=\tiny,
                    draw=none
                },
                title={}
            ]

            \addplot[name path=CE_Smin, draw=none, forget plot] coordinates {
                (1, 0.000) (2, 0.000) (3, 0.000)
            };
            \addplot[name path=CE_Smax, draw=none, forget plot] coordinates {
                (1, 0.133) (2, 0.316) (3, 0.220)
            };
            \addplot[blue!20, opacity=0.5, area legend] fill between[of=CE_Smin and CE_Smax];
            \addlegendentry{Range $\operatorname{CE}(\Y_S)$}

            \addplot[color=blue, thick, mark=*] coordinates {
                (1, 0.041) (2, 0.071) (3, 0.08)
            };
            \addlegendentry{Mean $\operatorname{CE}(\Y_S)$}

            \addplot[name path=CE_ESmin, draw=none, forget plot] coordinates {
                (1, 0.000) (2, 0.000) (3, 0.234)
            };
            \addplot[name path=CE_ESmax, draw=none, forget plot] coordinates {
                (1, 0.500) (2, 0.519) (3, 0.500)
            };
            \addplot[orange!30, opacity=0.5, area legend] fill between[of=CE_ESmin and CE_ESmax];
            \addlegendentry{Range $\operatorname{CE}(\Y_{ES})$}

            \addplot[color=orange, thick, mark=square*] coordinates {
                (1, 0.371) (2, 0.389) (3, 0.427)
            };
            \addlegendentry{Mean $\operatorname{CE}(\Y_{ES})$}

            \end{axis}
        \end{tikzpicture}
        \caption{coverage error for~\Cref{tab_dimensions}.}
        \label{fig:ce_dimensions}
    \end{minipage}
\end{figure}

Indeed, the quality advantage of $\Y_S$ persists in higher dimensions. Both, the ratios of supported non-dominated to all non-dominated points and the hypervolume ratios show only minimal variation. For $p=5$, $\Y_S$ covers $99.8\%$ of the hypervolume, whereas $\Y_{ES}$ covers only $64\%$. While there is a slight increase in the coverage error for the sets $\Y_S$ with increasing dimensions, it remains relatively small, with a mean value of just $0.080$ even in the case of $p = 5$.  The coverage error of $\Y_{ES}$ ($0.427$) is more than five times higher than that of $\Y_S$ ($0.080$).  Similarly, the $\varepsilon$-indicator emphasizes this performance gap. While $\Y_S$ maintains a near-optimal average of $1.002$ across all dimensions, the average excess error of $\Y_{ES}$ (above $1.0$) is more than $26$ times larger than that of $\Y_S$. These observations suggest that the effectiveness of supported non-dominated points as high-quality representations persists even in higher-dimensional settings.

\paragraph{High Number of Supply and Demand Nodes} 
The final test class consists of NETGEN benchmark instances constructed using parameters based on the test sets N01, N04, N07, and N12 from~\cite{raith09}.
As detailed in \Cref{tab_values_demand}, these instances feature a larger supply value and a greater number of source and sink nodes than the previous test cases. Consequently, due to higher  supply and fewer transshipment nodes, the ratio of supported to total non-dominated points ($|\Y_S|/|\Y_N|$) is expected to be lower. For all instances, the parameters $\text{\texttt{mincost}}=0$, $\text{\texttt{maxcost}}=100$, $\text{\texttt{capacitated}}=100$, $\text{\texttt{mincap}}=0$, and $\text{\texttt{maxcap}}=50$ were fixed, generating 15 instances per class.

\Cref{tab_demand} and~\Cref{tab_demand_3} present the results for the bi-objective and tri-objective cases, respectively. Notably, for $p=3$, we were only able to generate results for classes N01 and N04. In the test classes N07 and N12, the number of non-dominated points increases drastically, making the computation too demanding in our experimental setting.

\begin{table}[htb]\footnotesize
	\centering
	\caption{NETGEN parameters for the different classes with more supply, sources and sinks, given as in~\cite{raith09}.}\label{tab_values_demand}
	\vspace{0.5cm}
\begin{tabular}{lccccccc}
	\toprule Name & $n$ & $m$ & \texttt{Sources} & \texttt{Sinks} & \texttt{Supply} & \parbox{4.5em}{\texttt{Transship.\newline sources}} & \parbox{4.5em}{\texttt{Transship. sinks}} \\
	\midrule N01 & 20 & 60 & 9 & 7 & 90 & 4 & 3 \\
	\midrule N04 & 40 & 120 & 18 & 14 & 180 & 9 & 7 \\
	\midrule N07 & 60 & 180 & 27 & 21 & 270& 14 & 10 \\
	\midrule N12 & 80 & 400 & 35 & 38 & 350 & 17 & 14 \\
	\bottomrule
\end{tabular}
\end{table}

\begin{table}[htb]\tiny
    \centering
    \caption{Numerical results for bi-objective instances with more supply, sinks, and sources. The classes are similar to the one from~\citet{raith09}. The classes again contain 15 instances.}\label{tab_demand}
    \vspace{0.4em}
    \begin{tabular}{ll r c c c c c c c c}\toprule
        & & & \multicolumn{2}{c}{Size {(Ratio)}} & \multicolumn{2}{c}{HVR $\uparrow$}  & \multicolumn{2}{c}{CE $\downarrow$} & \multicolumn{2}{c}{$I_\varepsilon$ $\downarrow$} \\
        \cmidrule(lr){4-5} \cmidrule(lr){6-7} \cmidrule(lr){8-9} \cmidrule(lr){10-11}
        Class & & \multicolumn{1}{c}{$|\Y_N|$} & $\Y_S$ & $\Y_{ES}$ & $\Y_S$ & $\Y_{ES}$ & $\Y_S$ & $\Y_{ES}$ & $\Y_S$ & $\Y_{ES}$ \\
        \midrule 
        
        N01     & min  & 29.0  & 0.068  & 0.02   & 0.948 & 0.805 & 0.023 & 0.098 & 1.002 & 1.01 \\
        & max  & 353.0  & 0.655  & 0.31  & 0.999 & 0.98 & 0.084 & 0.361 & 1.021 & 1.082 \\
        & mean   & 162.867  & 0.258  & 0.098  & 0.988 & 0.918 & 0.043 & 0.174 & 1.009 & 1.042 \\

        \midrule
        
        N04 & min  &  120.000 & 0.103 & 0.033 & 0.990 & 0.933 & 0.021 & 0.068 & 1.002 & 1.012 \\
            & max  &  819.000 & 0.271 & 0.133 & 0.997 & 0.982 & 0.051 & 0.274 & 1.010 & 1.035 \\
            & mean &  388.533 & 0.175 & 0.069 & 0.993 & 0.965 & 0.029 & 0.115 & 1.006 & 1.022 \\
        \midrule
        
       N07 & min  &  452.000 & 0.080 & 0.030 & 0.993 & 0.973 & 0.018 & 0.048 & 1.003 & 1.010 \\
            & max  & 1455.000 & 0.195 & 0.086 & 0.996 & 0.987 & 0.030 & 0.124 & 1.008 & 1.030 \\
            & mean &  736.067 & 0.123 & 0.058 & 0.995 & 0.982 & 0.023 & 0.073 & 1.006 & 1.018 \\
        \midrule
        
        N12 & min  & 1192.000 & 0.063 & 0.034 & 0.995 & 0.984 & 0.010 & 0.028 & 1.004 & 1.009 \\
            & max  & 2577.000 & 0.099 & 0.059 & 0.998 & 0.995 & 0.023 & 0.060 & 1.007 & 1.033 \\
            & mean & 1895.533 & 0.076 & 0.043 & 0.997 & 0.992 & 0.017 & 0.045 & 1.006 & 1.016 \\
        \midrule
        
        $\boldsymbol{\varnothing}$ & \textbf{mean} & 795.750 & 0.158 & 0.067 & 0.993 & 0.964 & 0.028 & 0.102 & 1.0070 & 1.025 \\
        \bottomrule
    \end{tabular}
\end{table}

\begin{figure}[ht]
    \centering
    \begin{minipage}[t]{0.48\textwidth}
        \centering
        \begin{tikzpicture}
            \begin{axis}[
                width=\linewidth,
                height=6cm,
                xlabel={},
                ylabel={},
                ymin=0.0, ymax=1.1, 
                hide obscured x ticks=false,
                xtick={1, 2, 3, 4},
                xticklabels={N01, N04, N07, N12},
                x tick label style={font=\footnotesize},
                grid=major,
                legend style={
                    at={(0.5, -0.2)},
                    anchor=north,
                    legend columns=2,
                    font=\tiny,
                    draw=none
                },
                title={}
            ]

            \addplot[color=blue, thick, dashed, mark=*] coordinates {
                (1, 0.258) (2, 0.175) (3, 0.123) (4, 0.076)
            };
            \addlegendentry{Ratio $\frac{|\Y_S|}{|\Y_N|}$}

            \addplot[color=blue, thick, mark=*] coordinates {
                (1, 0.988) (2, 0.993) (3, 0.995) (4, 0.997)
            };
            \addlegendentry{HVR($\Y_S$)}

            \addplot[color=orange, thick, dashed, mark=square*] coordinates {
                (1, 0.098) (2, 0.069) (3, 0.058) (4, 0.043)
            };
            \addlegendentry{Ratio $\frac{|\Y_{ES}|}{|\Y_N|}$}

            \addplot[color=orange, thick, mark=square*] coordinates {
                (1, 0.918) (2, 0.965) (3, 0.982) (4, 0.992)
            };
            \addlegendentry{HVR($\Y_{ES}$)}

            \end{axis}
        \end{tikzpicture}
        \caption{Ratios \& HVR, increased supply}
        \label{fig:ratio_hvr_supply}
    \end{minipage}
    \hfill
    \begin{minipage}[t]{0.48\textwidth}
        \centering
        \begin{tikzpicture}
            \begin{axis}[
                width=\linewidth,
                height=6cm,
                xlabel={},
                ylabel={},
                ymin=0.0, ymax=0.45,
                hide obscured x ticks=false,
                xtick={1, 2, 3, 4},
                xticklabels={N01, N04, N07, N12},
                x tick label style={font=\footnotesize},
                grid=major,
                legend style={
                    at={(0.5, -0.25)},
                    anchor=north,
                    legend columns=2,
                    font=\tiny,
                    draw=none
                },
                title={}
            ]

            \addplot[name path=CE_Smin, draw=none, forget plot] coordinates {
                (1, 0.023) (2, 0.021) (3, 0.018) (4, 0.010)
            };
            \addplot[name path=CE_Smax, draw=none, forget plot] coordinates {
                (1, 0.084) (2, 0.051) (3, 0.030) (4, 0.023)
            };
            \addplot[blue!20, opacity=0.5, area legend] fill between[of=CE_Smin and CE_Smax];
            \addlegendentry{Range $\operatorname{CE}(\Y_S)$}

            \addplot[color=blue, thick, mark=*] coordinates {
                (1, 0.043) (2, 0.029) (3, 0.023) (4, 0.017)
            };
            \addlegendentry{Mean $\operatorname{CE}(\Y_S)$}

            \addplot[name path=CE_ESmin, draw=none, forget plot] coordinates {
                (1, 0.098) (2, 0.068) (3, 0.048) (4, 0.028)
            };
            \addplot[name path=CE_ESmax, draw=none, forget plot] coordinates {
                (1, 0.361) (2, 0.274) (3, 0.124) (4, 0.060)
            };
            \addplot[orange!30, opacity=0.5, area legend] fill between[of=CE_ESmin and CE_ESmax];
            \addlegendentry{Range $\operatorname{CE}(\Y_{ES})$}

            \addplot[color=orange, thick, mark=square*] coordinates {
                (1, 0.174) (2, 0.115) (3, 0.073) (4, 0.045)
            };
            \addlegendentry{Mean $\operatorname{CE}(\Y_{ES})$}

            \end{axis}
        \end{tikzpicture}
        \caption{coverage error}
        \label{fig:ce_supply}
    \end{minipage}
\end{figure}

\begin{table}[htb]\tiny
    \centering
    \caption{Numerical results for same instances as in~\Cref{tab_demand} but with 3 objectives.}\label{tab_demand_3}
    \vspace{0.4em}
    \begin{tabular}{ll r c c c c c c c c}
        \toprule
        & & & \multicolumn{2}{c}{Size {(Ratio)}} & \multicolumn{2}{c}{HVR $\uparrow$}  & \multicolumn{2}{c}{CE $\downarrow$} & \multicolumn{2}{c}{$I_\varepsilon$ $\downarrow$} \\
        \cmidrule(lr){4-5} \cmidrule(lr){6-7} \cmidrule(lr){8-9} \cmidrule(lr){10-11}
        Class & & \multicolumn{1}{c}{$|\Y_N|$} & $\Y_S$ & $\Y_{ES}$ & $\Y_S$ & $\Y_{ES}$ & $\Y_S$ & $\Y_{ES}$ & $\Y_S$ & $\Y_{ES}$ \\
        \midrule 
        
        N01 & min  &   512.000 & 0.075 & 0.009 & 0.983 & 0.828 & 0.072 & 0.122 & 1.007 & 1.029 \\
            & max  &  8443.000 & 0.311 & 0.055 & 0.998 & 0.971 & 0.262 & 0.262 & 1.017 & 1.103 \\
            & mean &  3647.467 & 0.154 & 0.022 & 0.991 & 0.915 & 0.119 & 0.209 & 1.011 & 1.056 \\
        \midrule
        
        N04 & min  &  2945.000 & 0.026 & 0.005 & 0.991 & 0.947 & 0.039 & 0.090 & 1.004 & 1.016 \\
            & max  & 88732.000 & 0.146 & 0.026 & 0.996 & 0.979 & 0.160 & 0.231 & 1.010 & 1.044 \\
            & mean & 37100.533 & 0.064 & 0.012 & 0.993 & 0.964 & 0.090 & 0.126 & 1.007 & 1.027 \\
        \midrule
        
        $\boldsymbol{\varnothing}$ & \textbf{mean} & 20374.000 & 0.109 & 0.017 & 0.992 & 0.940 & 0.105 & 0.168 & 1.009 & 1.042 \\
        \bottomrule
    \end{tabular}
\end{table}

Focusing on the bi-objective test classes for which all instances could be solved (\Cref{tab_demand}), the data confirms the presumed drop in the ratio of supported points. The mean proportion $|\Y_S|/|\Y_N|$ decreases from 0.258 in N01 to just 0.076 in N12. Remarkably, as the problem scales and $\Y_S$ becomes a fractionally smaller subset, its representation quality actually improves. \Cref{fig:ce_supply} illustrates that the mean coverage error for $\Y_S$ decreases from 0.043 (N01) to 0.017 (N12), maintaining an exceptional overall mean $\operatorname{HVR}$ of 0.993. In contrast, while the extreme supported set $\Y_{ES}$ achieves a respectable overall mean $\operatorname{HVR}$ of 0.964 in these instances, it remains distinctly outperformed by $\Y_S$. Across all classes, the average coverage error of $\Y_{ES}$ (0.102) is roughly 3.6 times larger than that of $\Y_S$ (0.028). The $\varepsilon$-indicator further underscores this gap: the average excess error above 1.0 for $\Y_{ES}$ (0.025) is more than 3.5 times larger than the excess error for $\Y_S$ (0.007).

Ultimately, these results confirm that even as problem complexity scales and the relative size of the supported set shrinks drastically, $\Y_S$ remains a highly robust, reliable, and superior representation compared to $\Y_{ES}$.

\paragraph{Additional Evaluations} Further experiments investigating the effects of capacity scaling, cost correlations, network density, and alternative graph dimensions yield consistent results. Detailed findings from these analyzes can be found in~\cite{koenenphd}, confirming the robust representation quality of $\mathcal{Y}_S$ across diverse network topologies. Additional instances with results can be found in the associated repository~\citep{koenen2026supprepresentations}.

\subsection{Summary of Results}

The numerical experiments establish a clear structural dichotomy among the considered network optimization problems. For binary and low-capacity network problems, our results are consistent with the observations in~\cite{Serpil2024}: the sets $\Y_S$ and $\Y_{ES}$ are nearly identical across all considered test instances, and, moreover, $\Y_{ES}$ provides a high-quality representation with respect to several quality indicators. 

For capacitated flow problems with larger arc capacities, however, we obtain different results. The set of supported non-dominated points is significantly larger than the set of extreme supported points, which can be explained by the underlying structural properties of such networks, as discussed in~\Cref{sec:NetworkProb}. Moreover, the representation quality of $\Y_{ES}$ deteriorates, whereas $\Y_S$ continues to provide high-quality representations across all considered instances.

These results confirm that the behavior observed for knapsack problems does not generalize to capacitated network flow problems, highlighting the benefits of using $\Y_S$ for effective decision support in flow-based applications. Notably, we also considered bounded knapsack instances with increased capacities, allowing items to be selected multiple times. In contrast to network flow problems, increasing capacities did not affect the relationship between $\Y_S$ and $\Y_{ES}$. The two sets remained nearly identical in size, and $\Y_{ES}$ continued to provide representations of comparable quality. This suggests that the substantial expansion of $\Y_S$ relative to $\Y_{ES}$ is not merely a consequence of larger integer domains, but rather a structural property.

Overall, our results suggest that supported non-dominated points $\Y_S$ consistently provide high-quality representations across all considered problem classes. While $\Y_{ES}$ performs well for binary and low-capacity problems, its quality deteriorates as network complexity increases. In contrast, $\Y_S$ maintains a high representation quality across a wide range of network structures, including sparse and dense topologies, varying capacity levels, and higher-dimensional instances.

\section{Quality of Fixed-Size Representations}\label{sec:FixedSize}

The cardinality of a representation is a crucial quality indicator in practical applications, which we have not yet taken into account. Large sets of alternative solutions are often unmanageable for decision makers (DMs), as evaluating hundreds of alternatives leads to choice paralysis. Consequently, DMs generally prefer a concise, fixed-size representation, typically ranging from $k=5$ to $k=20$ \citep{HUMBERTROPERS2025}.

\paragraph{Representation Approaches}
To construct fixed-size representations, the literature generally distinguishes between \emph{a~priori} and \emph{a posteriori} (or filtering) methods. A priori methods generate representative subsets of a set $S$ directly without first computing the entire set $S$ (e.\,g., \citealt{hamacher07,Sylva2007,eusebio14,kidd2020,Kirlik2025}). As these methods construct the representation (often incrementally) without full knowledge of the objective space, they generally cannot guarantee the selection of an optimal subset with respect to specific quality metrics. Indeed, selecting a representative subset of fixed size $k$ that maximizes the hypervolume indicator is theoretically possible without full knowledge of $\Y_N$, but it is often computationally too expensive in practice.

Filtering algorithms, in contrast, guarantee optimality by solving a \emph{subset selection problem} that identifies an optimal subset of size $k$ from a known candidate set $S$ with respect to a specific indicator (e.\,g., minimizing coverage error). \cite{vaz15representation} propose several polynomial-time algorithms to find optimal subsets of $\Y_N$ for coverage, uniformity, and the $\varepsilon$-indicator. Similar approaches exist for the hypervolume subset selection problem \citep{kuhn2016,gomes2018,Groz2019,guerreiro2021}.
A survey of these methods is provided in~\cite{guerreiro2021}. In the following, we focus on approaches where only a single quality indicator is optimized.

\paragraph{A Two-Phase Selection Strategy} 

While filtering algorithms yield optimal subsets with respect to a specific quality indicator, they require generating the complete, often intractable set $\Y_{N}$ in advance, only to discard the vast majority of its points during the filtering step. This process is computationally demanding, both  in generating $\Y_N$ itself and in solving the subset selection problem on a potentially massive input set.

To overcome this computational burden, we propose a new two-phase strategy. Instead of generating and filtering the entire set $\Y_{N}$, we construct representations by directly selecting $k$ points from the (extreme) supported non-dominated points. As demonstrated in \Cref{sec:RepQualSuppNonPoint}, these sets consistently yield high-quality representations of $\Y_{N}$, depending on the given problem structure. However, their cardinality often exceeds a manageable number of solutions.

Therefore, rather than presenting them directly to the DM or filtering the full set $\Y_N$, we  use $\Y_S$ or $\Y_{ES}$ as candidate sets for filtering, which offers the following two advantages: First, the (extreme) supported points are typically much easier to compute than the full set $\Y_{N}$. Second, the sizes of the sets $\Y_S$ or $\Y_{ES}$ are significantly smaller. As a result, the subsequent filtering step requires only a fraction of the computational effort. 

The core research question we address in this section is therefore: \emph{How much representation quality is lost when restricting the candidate set to $\Y_S$ or $\Y_{ES}$?}

\subsection{Problem Definition and Relative Quality Metrics.}

The subset selection problem involves finding a subset ${R}^*_k(S) \subseteq S$ of size $k$ from a candidate set $S$ that optimizes a specific quality indicator $I$. Depending on the quality measure applied, this takes different forms: 

\begin{equation*}
    {R}^*_k(S) = 
    \begin{cases} 
        \underset{{R} \subseteq S, \; |{R}| = k}{\arg\max} \; I({R}) & \text{for maximization metrics (e.\,g., uniformity)} \\[15pt]
        \underset{{R} \subseteq S,|{R}| = k}{\arg\min} \; \; I({R}) & \text{for minimization metrics (e.\,g., coverage error)}
    \end{cases}
\end{equation*}

To quantify the trade-off between the quality of the subset selected from the restricted set $S$ and the optimal subset chosen from the  entire space $\Y_N$, we define  the \emph{Relative Quality Ratio (RQR)}:
\begin{equation*}
    \operatorname{RQR}_k^I(S) = 
    \begin{cases} 
        \frac{I({R}^*_k(S))}{I({R}^*_k{(\Y_N}))} & \text{for maximization metrics (e.\,g., uniformity)} \\[10pt]
        \frac{I({R}^*_k{(\Y_N}))}{I({R}^*_k{(S))}} & \text{for minimization metrics (e.\,g., coverage error)}
    \end{cases}
\end{equation*}

Because $S \subseteq \Y_{N}$, the RQR is strictly bounded between 0 and 1. An RQR value approaching 1.0 indicates that the optimal $k$-element subset drawn from the restricted set $S$ performs nearly as well as the best possible subset drawn from the entire space $\Y_{N}$. This metric allows us to precisely quantify the representation quality retained when restricting the candidate set to $\Y_{S}$ or $\Y_{ES}$. Ultimately, a high RQR proves that the restricted search space contains a representation that is practically indistinguishable from the theoretical optimum.

\subsection{Experimental Setup}

To validate the proposed two-phase selection strategy, we conducted numerical experiments using bi-objective network optimization instances described in~\Cref{sec:RepQualSuppNonPoint}. To identify optimal subsets $R^*_k(S)$ with respect to different quality measures, we utilized the open-source implementation by~\cite{moqm}, which relies on the exact dynamic programming approaches proposed by~\cite{vaz15representation}.

We determined optimal subsets with respect to coverage, uniformity, hypervolume, and the $\varepsilon$-indicator, drawn from candidate sets $S \in \{\Y_{ES}, \Y_{S}, \Y_N\}$ for sizes $k \in \{5,10,15,20\}$. We then computed the RQR to compare the fidelity of subsets chosen from the (extreme) supported points against those chosen from the entire set $\Y_N$.

The tables and figures below present results for the following representative test classes: capacitated minimum cost flow instances (N07, N12); a bi-objective~\ref{eq:MOIMCF} instance class featuring a small number of supply and sink nodes but a higher number of nodes, arcs, and cost coefficients (NG\_400\_2500\_c50, $n=400$ nodes, $m=2500$ arcs, cost coefficients $c^1, c^2 \in [1,50]$, with the remaining NETGEN parameters as described in~\Cref{sec:4.1}); and a binary shortest path instance (SP\_500\_1300\_c100, $n=500$ nodes, $m=1300$ arcs, cost coefficients $c^1, c^2 \in [1,100]$). Results for multi-objective minimum spanning tree and bi-objective cable trench test classes closely mirror the behavior of~\ref{eq:MOSP} test instances and are omitted for simplicity, though full data is available in the associated repository~\citep{koenen2026supprepresentations}.

\begin{figure}[ht]
    \centering
    \begin{minipage}[c]{0.45\textwidth}
        \centering
           \captionof{table}{Quality Metrics over $k$ }\label{tab:metrics_side}
        \resizebox{1.1\textwidth}{!}{
        \tiny
    \begin{tabular}{l cc  cc  cc  cc}
        \toprule
        & \multicolumn{2}{c}{\textbf{RQR$_k^{\operatorname{HV}}$}} & \multicolumn{2}{c}{\textbf{RQR$_k^{\operatorname{CE}}$}} & \multicolumn{2}{c}{\textbf{RQR$_k^{I_\varepsilon}$}} & \multicolumn{2}{c}{\textbf{RQR$_k^{\operatorname{UN}}$}} \\
        \cmidrule(lr){2-3} \cmidrule(lr){4-5} \cmidrule(lr){6-7} \cmidrule(l){8-9}
        $k$ & $\mathbf{\Y_S}$ & $\mathbf{\Y_{ES}}$ & $\mathbf{\Y_S}$ & $\mathbf{\Y_{ES}}$ & $\mathbf{\Y_S}$ & $\mathbf{\Y_{ES}}$ & $\mathbf{\Y_S}$ & $\mathbf{\Y_{ES}}$ \\
        \midrule
        $5$  & 1.000 & 0.999 & 0.945 & 0.781 & 0.998 & 0.990 & 0.983 & 0.915 \\
        $10$ & 1.000 & 0.999 & 0.872 & 0.647 & 0.997 & 0.991 & 0.943 & 0.820 \\
        $15$ & 1.000 & 0.998 & 0.792 & 0.541 & 0.997 & 0.989 & 0.911 & 0.727 \\
        $20$ & 1.000 & 0.996 & 0.753 & 0.484 & 0.997 & 0.988 & 0.876 & 0.630 \\
        \midrule
        $\boldsymbol{\varnothing}$ & 1.000 & 0.998 & 0.841 & 0.613 & 0.997 & 0.990 & 0.928 & 0.773 \\
        \bottomrule
    \end{tabular}}
    \vspace{1.5em}
            \captionof{table}{Instance N07}
        \label{tab:instance_sizes}
        \resizebox{0.5\textwidth}{!}{
            \begin{tabular}{lrrr}
            \toprule
            Stat & $|\Y_N|$ & $|\Y_S|$ & $|\Y_E|$ \\
            \midrule
            min & 452 & 66 & 33 \\
            max & 1455 & 117 & 49 \\
            mean & 736.07 & 85.20 & 40.07 \\
            \bottomrule
            \end{tabular}
        }

    \end{minipage}
    \hfill 
    \begin{minipage}[c]{0.5\textwidth}
        \centering

        \begin{tikzpicture}
            \begin{axis}[
                width=\linewidth,
                height=7cm,
                xlabel={}, 
                ylabel={},
                ymin=0.0, ymax=1.0,
                hide obscured x ticks=false,
                xtick={5, 10, 15, 20},
                xticklabels={$k=5$, $k=10$, $k=15$, $k=20$},
                grid=major,
                legend style={
                    at={(0.5, 1.03)}, 
                    anchor=south,    
                    legend columns=2, 
                    font=\tiny,
                    draw=none,
                    /tikz/every even column/.append style={column sep=0.3cm}
                },
                title style={yshift=1.5ex},
                title={}
            ]

            \addplot[name path=Smin, draw=none, forget plot] coordinates {
                (5, 0.9086) (10, 0.7753) (15, 0.6827) (20, 0.6342)
            };
            \addplot[name path=Smax, draw=none, forget plot] coordinates {
                (5, 0.9817) (10, 0.9328) (15, 0.8805) (20, 0.8821)
            };
            
            \addplot[blue!20, opacity=0.5, area legend] fill between[of=Smin and Smax];
            \addlegendentry{Range $\operatorname{RQR}^{\operatorname{CE}}_k(\Y_{S})$}
            
            \addplot[color=blue, thick, mark=*] coordinates {
                (5, 0.9447) (10, 0.8722) (15, 0.7921) (20, 0.7534)
            };
            \addlegendentry{Mean $\operatorname{RQR}^{\operatorname{CE}}_k(\Y_{S})$}

            \addplot[name path=ESmin, draw=none, forget plot] coordinates {
                (5, 0.6386) (10, 0.4534) (15, 0.4270) (20, 0.3114)
            };
            \addplot[name path=ESmax, draw=none, forget plot] coordinates {
                (5, 0.9371) (10, 0.8444) (15, 0.7276) (20, 0.7472)
            };
            
            \addplot[orange!30, opacity=0.5, area legend] fill between[of=ESmin and ESmax];
            \addlegendentry{Range$\operatorname{RQR}^{\operatorname{CE}}_k(\Y_{ES})$}

            \addplot[color=orange, thick, mark=square*] coordinates {
                (5, 0.7808) (10, 0.6472) (15, 0.5406) (20, 0.4842)
            };
            \addlegendentry{Mean $\operatorname{RQR}^{\operatorname{CE}}_k(\Y_{ES})$}

            \end{axis}
        \end{tikzpicture}
        \captionof{figure}{Comparison of RQR metrics}
        \label{fig:metrics_plot}
    \end{minipage}
\end{figure}

\begin{figure}[ht]
    \centering
    \begin{minipage}[c]{0.45\textwidth}
        \centering
        \captionof{table}{Quality Metrics over $k$}
        \label{tab:metrics_side_n12}
        \resizebox{1.1\textwidth}{!}{
        \tiny
        \begin{tabular}{l cc  cc  cc  cc}
            \toprule
            & \multicolumn{2}{c}{\textbf{RQR$_k^{\operatorname{HV}}$}} & \multicolumn{2}{c}{\textbf{RQR$_k^{\operatorname{CE}}$}} & \multicolumn{2}{c}{\textbf{RQR$_k^{I_\varepsilon}$}} & \multicolumn{2}{c}{\textbf{RQR$_k^{\operatorname{UN}}$}} \\
            \cmidrule(lr){2-3} \cmidrule(lr){4-5} \cmidrule(lr){6-7} \cmidrule(l){8-9}
            $k$ & $\mathbf{\mathcal{Y}_S}$ & $\mathbf{\mathcal{Y}_{ES}}$ & $\mathbf{\mathcal{Y}_S}$ & $\mathbf{\mathcal{Y}_{ES}}$ & $\mathbf{\mathcal{Y}_S}$ & $\mathbf{\mathcal{Y}_{ES}}$ & $\mathbf{\mathcal{Y}_S}$ & $\mathbf{\mathcal{Y}_{ES}}$ \\
            \midrule
            $5$  & 0.924 & 0.924 & 0.954 & 0.892 & 0.997 & 0.992 & 0.988 & 0.977 \\
            $10$ & 0.962 & 0.962 & 0.900 & 0.763 & 0.997 & 0.992 & 0.969 & 0.920 \\
            $15$ & 0.975 & 0.975 & 0.853 & 0.677 & 0.997 & 0.991 & 0.939 & 0.864 \\
            $20$ & 0.982 & 0.981 & 0.814 & 0.600 & 0.997 & 0.991 & 0.924 & 0.813 \\
            \midrule
            $\boldsymbol{\varnothing}$ & 0.961 & 0.960 & 0.880 & 0.733 & 0.997 & 0.991 & 0.955 & 0.893 \\
            \bottomrule
        \end{tabular}}
        \vspace{1.5em}
        \captionof{table}{Instance N12}
        \label{tab:instance_sizes_n12}
        \resizebox{0.5\textwidth}{!}{
            \begin{tabular}{lrrr}
            \toprule
            Stat & $|\mathcal{Y}_N|$ & $|\mathcal{Y}_S|$ & $|\mathcal{Y}_E|$ \\
            \midrule
            min & 1192 & 96 & 65 \\
            max & 2577 & 188 & 95 \\
            mean & 1895.53 & 142.87 & 79.67 \\
            \bottomrule
            \end{tabular}
        }
    \end{minipage}
    \hfill 
    \begin{minipage}[c]{0.5\textwidth}
        \centering
        
        \begin{tikzpicture}
            \begin{axis}[
                width=\linewidth,
                height=7cm,
                xlabel={}, 
                ylabel={},
                ymin=0.0, ymax=1.0,
                hide obscured x ticks=false,
                xtick={5, 10, 15, 20},
                xticklabels={$k=5$, $k=10$, $k=15$, $k=20$},
                grid=major,
                legend style={
                    at={(0.5, 1.03)}, 
                    anchor=south,    
                    legend columns=2, 
                    font=\tiny,
                    draw=none,
                    /tikz/every even column/.append style={column sep=0.3cm}
                },
                title style={yshift=1.5ex},
                title={}
            ]
            \addplot[name path=Smin, draw=none, forget plot] coordinates {
                (5, 0.9181) (10, 0.8154) (15, 0.7696) (20, 0.6954)
            };
            \addplot[name path=Smax, draw=none, forget plot] coordinates {
                (5, 0.9892) (10, 0.9434) (15, 0.9115) (20, 0.8805)
            };
            \addplot[blue!20, opacity=0.5, area legend] fill between[of=Smin and Smax];
            \addlegendentry{Range $\operatorname{RQR}^{\operatorname{CE}}_k(\mathcal{Y}_{S})$}
            \addplot[color=blue, thick, mark=*] coordinates {
                (5, 0.9543) (10, 0.8995) (15, 0.8528) (20, 0.8136)
            };
            \addlegendentry{Mean $\operatorname{RQR}^{\operatorname{CE}}_k(\mathcal{Y}_{S})$}

            \addplot[name path=ESmin, draw=none, forget plot] coordinates {
                (5, 0.8069) (10, 0.6143) (15, 0.5749) (20, 0.4963)
            };
            \addplot[name path=ESmax, draw=none, forget plot] coordinates {
                (5, 0.9606) (10, 0.8660) (15, 0.8114) (20, 0.7273)
            };
            \addplot[orange!30, opacity=0.5, area legend] fill between[of=ESmin and ESmax];
            \addlegendentry{Range $\operatorname{RQR}^{\operatorname{CE}}_k(\mathcal{Y}_{ES})$}
            \addplot[color=orange, thick, mark=square*] coordinates {
                (5, 0.8920) (10, 0.7631) (15, 0.6771) (20, 0.6004)
            };
            \addlegendentry{Mean $\operatorname{RQR}^{\operatorname{CE}}_k(\mathcal{Y}_{ES})$}
            \end{axis}
        \end{tikzpicture}
        \captionof{figure}{Comparison of RQR metrics}
        \label{fig:metrics_plot_n12}
    \end{minipage}
\end{figure}

\begin{figure}[ht]
    \centering
    \begin{minipage}[c]{0.45\textwidth}
        \centering
        \captionof{table}{Quality Metrics over $k$}
        \label{tab:metrics_side_400_2500_supp}
        \resizebox{1.1\textwidth}{!}{
        \tiny
        \begin{tabular}{l cc  cc  cc  cc}
            \toprule
            & \multicolumn{2}{c}{\textbf{RQR$_k^{\operatorname{HV}}$}} & \multicolumn{2}{c}{\textbf{RQR$_k^{\operatorname{CE}}$}} & \multicolumn{2}{c}{\textbf{RQR$_k^{I_\varepsilon}$}} & \multicolumn{2}{c}{\textbf{RQR$_k^{\operatorname{UN}}$}} \\
            \cmidrule(lr){2-3} \cmidrule(lr){4-5} \cmidrule(lr){6-7} \cmidrule(l){8-9}
            $k$ & $\mathbf{\mathcal{Y}_S}$ & $\mathbf{\mathcal{Y}_{ES}}$ & $\mathbf{\mathcal{Y}_S}$ & $\mathbf{\mathcal{Y}_{ES}}$ & $\mathbf{\mathcal{Y}_S}$ & $\mathbf{\mathcal{Y}_{ES}}$ & $\mathbf{\mathcal{Y}_S}$ & $\mathbf{\mathcal{Y}_{ES}}$ \\
            \midrule
            $5$  & 0.930 & 0.930 & 0.959 & 0.730 & 0.997 & 0.975 & 0.988 & 0.910 \\
            $10$ & 0.966 & 0.963 & 0.893 & 0.599 & 0.997 & 0.973 & 0.962 & 0.745 \\
            $15$ & 0.978 & 0.972 & 0.834 & 0.460 & 0.997 & 0.969 & 0.942 & 0.545 \\
            $20$ & 0.984 & 0.975 & 0.798 & 0.361 & 0.997 & 0.966 & 0.928 & 0.402 \\
            \midrule
            $\boldsymbol{\varnothing}$ & 0.964 & 0.960 & 0.871 & 0.538 & 0.997 & 0.971 & 0.955 & 0.651 \\
            \bottomrule
        \end{tabular}}
        \vspace{1.5em}
        \captionof{table}{Instance NG\_400\_2500\_c50}
        \label{tab:instance_sizes_400_2500_supp}
        \resizebox{0.5\textwidth}{!}{
            \begin{tabular}{lrrr}
            \toprule
            Stat & $|\mathcal{Y}_N|$ & $|\mathcal{Y}_S|$ & $|\mathcal{Y}_E|$ \\
            \midrule
            min & 359 & 65 & 23 \\
            max & 1223 & 145 & 42 \\
            mean & 723.33 & 106.33 & 29.60 \\
            \bottomrule
            \end{tabular}
        }
    \end{minipage}
    \hfill 
    \begin{minipage}[c]{0.5\textwidth}
        \centering
        
        \begin{tikzpicture}
            \begin{axis}[
                width=\linewidth,
                height=7cm,
                xlabel={}, 
                ylabel={},
                ymin=0.0, ymax=1.0,
                hide obscured x ticks=false,
                xtick={5, 10, 15, 20},
                xticklabels={$k=5$, $k=10$, $k=15$, $k=20$},
                grid=major,
                legend style={
                    at={(0.5, 1.03)}, 
                    anchor=south,    
                    legend columns=2, 
                    font=\tiny,
                    draw=none,
                    /tikz/every even column/.append style={column sep=0.3cm}
                },
                title style={yshift=1.5ex},
                title={}
            ]

            \addplot[name path=Smin, draw=none, forget plot] coordinates {
                (5, 0.9310) (10, 0.7837) (15, 0.6444) (20, 0.6866)
            };
            \addplot[name path=Smax, draw=none, forget plot] coordinates {
                (5, 0.9895) (10, 0.9623) (15, 0.9076) (20, 0.9375)
            };
            
            \addplot[blue!20, opacity=0.5, area legend] fill between[of=Smin and Smax];
            \addlegendentry{Range $\operatorname{RQR}^{\operatorname{CE}}_k(\mathcal{Y}_{S})$}

            \addplot[color=blue, thick, mark=*] coordinates {
                (5, 0.9559) (10, 0.8836) (15, 0.8219) (20, 0.8045)
            };
            \addlegendentry{Mean $\operatorname{RQR}^{\operatorname{CE}}_k(\mathcal{Y}_{S})$}

            \addplot[name path=ESmin, draw=none, forget plot] coordinates {
                (5, 0.5618) (10, 0.4731) (15, 0.3242) (20, 0.2365)
            };
            \addplot[name path=ESmax, draw=none, forget plot] coordinates {
                (5, 0.8063) (10, 0.7942) (15, 0.5773) (20, 0.5108)
            };
            
            \addplot[orange!30, opacity=0.5, area legend] fill between[of=ESmin and ESmax];
            \addlegendentry{Range $\operatorname{RQR}^{\operatorname{CE}}_k(\mathcal{Y}_{ES})$}

            \addplot[color=orange, thick, mark=square*] coordinates {
                (5, 0.7130) (10, 0.5692) (15, 0.4442) (20, 0.3485)
            };
            \addlegendentry{Mean $\operatorname{RQR}^{\operatorname{CE}}_k(\mathcal{Y}_{ES})$}

            \end{axis}
        \end{tikzpicture}
        \captionof{figure}{Comparison of RQR metrics}
        \label{fig:metrics_plot_400_2500_supp}
    \end{minipage}
\end{figure}

\begin{figure}[ht]
    \centering
    \begin{minipage}[c]{0.45\textwidth}
        \centering
        \captionof{table}{Quality Metrics over $k$}
        \label{tab:metrics_side_sp_supp_v2}
        \resizebox{1.1\textwidth}{!}{
        \tiny
        \begin{tabular}{l cc  cc  cc  cc}
            \toprule
            & \multicolumn{2}{c}{\textbf{RQR$_k^{\operatorname{HV}}$}} & \multicolumn{2}{c}{\textbf{RQR$_k^{\operatorname{CE}}$}} & \multicolumn{2}{c}{\textbf{RQR$_k^{I_\varepsilon}$}} & \multicolumn{2}{c}{\textbf{RQR$_k^{\operatorname{UN}}$}} \\
            \cmidrule(lr){2-3} \cmidrule(lr){4-5} \cmidrule(lr){6-7} \cmidrule(l){8-9}
            $k$ & $\mathbf{\mathcal{Y}_S}$ & $\mathbf{\mathcal{Y}_{ES}}$ & $\mathbf{\mathcal{Y}_S}$ & $\mathbf{\mathcal{Y}_{ES}}$ & $\mathbf{\mathcal{Y}_S}$ & $\mathbf{\mathcal{Y}_{ES}}$ & $\mathbf{\mathcal{Y}_S}$ & $\mathbf{\mathcal{Y}_{ES}}$ \\
            \midrule
            $5$  & 0.926 & 0.926 & 0.898 & 0.898 & 0.996 & 0.996 & 0.972 & 0.970 \\
            $10$ & 0.964 & 0.964 & 0.788 & 0.790 & 0.995 & 0.995 & 0.916 & 0.916 \\
            $15$ & 0.977 & 0.977 & 0.709 & 0.707 & 0.996 & 0.995 & 0.840 & 0.839 \\
            $20$ & 0.984 & 0.984 & 0.657 & 0.651 & 0.996 & 0.995 & 0.788 & 0.785 \\
            \midrule
            $\boldsymbol{\varnothing}$ & 0.963 & 0.963 & 0.763 & 0.762 & 0.996 & 0.995 & 0.879 & 0.878 \\
            \bottomrule
        \end{tabular}}
        \vspace{1.5em}
        \captionof{table}{Instance SP\_500\_1300\_c100}
        \label{tab:instance_sizes_sp_supp_v2}
        \resizebox{0.5\textwidth}{!}{
            \begin{tabular}{lrrr}
            \toprule
            Stat & $|\mathcal{Y}_N|$ & $|\mathcal{Y}_S|$ & $|\mathcal{Y}_E|$ \\
            \midrule
            min & 182 & 36 & 35 \\
            max & 531 & 58 & 58 \\
            mean & 343.20 & 45.53 & 44.87 \\
            \bottomrule
            \end{tabular}
        }
    \end{minipage}
    \hfill 
    \begin{minipage}[c]{0.5\textwidth}
        \centering
        
        \begin{tikzpicture}
            \begin{axis}[
                width=\linewidth,
                height=7cm,
                xlabel={}, 
                ylabel={},
                ymin=0.0, ymax=1.05, 
                hide obscured x ticks=false,
                xtick={5, 10, 15, 20},
                xticklabels={$k=5$, $k=10$, $k=15$, $k=20$},
                grid=major,
                legend style={
                    at={(0.5, 1.03)}, 
                    anchor=south,    
                    legend columns=2, 
                    font=\tiny,
                    draw=none,
                    /tikz/every even column/.append style={column sep=0.3cm}
                },
                title style={yshift=1.5ex},
                title={}
            ]

            \addplot[name path=Smin, draw=none, forget plot] coordinates {
                (5, 0.7762) (10, 0.6443) (15, 0.6311) (20, 0.5412)
            };
            \addplot[name path=Smax, draw=none, forget plot] coordinates {
                (5, 0.9575) (10, 0.9125) (15, 0.8679) (20, 0.9091)
            };
            \addplot[blue!20, opacity=0.5, area legend] fill between[of=Smin and Smax];
            \addlegendentry{Range $\operatorname{RQR}^{\operatorname{CE}}_k(\mathcal{Y}_{S})$}
            
            \addplot[color=blue, thick, mark=*] coordinates {
                (5, 0.8984) (10, 0.7876) (15, 0.7092) (20, 0.6569)
            };
            \addlegendentry{Mean $\operatorname{RQR}^{\operatorname{CE}}_k(\mathcal{Y}_{S})$}

            \addplot[name path=ESmin, draw=none, forget plot] coordinates {
                (5, 0.7735) (10, 0.6443) (15, 0.6108) (20, 0.4919)
            };
            \addplot[name path=ESmax, draw=none, forget plot] coordinates {
                (5, 0.9489) (10, 0.9312) (15, 0.8679) (20, 0.9091)
            };
            \addplot[orange!30, opacity=0.5, area legend] fill between[of=ESmin and ESmax];
            \addlegendentry{Range $\operatorname{RQR}^{\operatorname{CE}}_k(\mathcal{Y}_{ES})$}

            \addplot[color=orange, thick, mark=square*] coordinates {
                (5, 0.8981) (10, 0.7899) (15, 0.7072) (20, 0.6514)
            };
            \addlegendentry{Mean $\operatorname{RQR}^{\operatorname{CE}}_k(\mathcal{Y}_{ES})$}
            \end{axis}
        \end{tikzpicture}
        \captionof{figure}{Comparison of RQR metrics}
        \label{fig:metrics_plot_sp_supp_v2}
    \end{minipage}
\end{figure}

\subsection{Discussion of Results}

The numerical results confirm that selecting $k$ representative points from the restricted sets $\Y_S$ or $\Y_{ES}$ yields near-optimal representations, successfully bypassing the computational burden of generating the full set $\Y_N$.

\paragraph{Structural Differences in Optimality Retention}
The viability of the candidate set is strongly dictated by the underlying problem structure. For capacitated network flow problems (e.\,g., classes N07, N12, and NG\_400\_2500\_c50), a substantial performance gap exists. In the NG\_400\_2500\_c50 class (\Cref{tab:metrics_side_400_2500_supp}), the overall mean $\operatorname{RQR}^{\operatorname{CE}}$ is 0.871 for $\Y_S$, whereas it drops to 0.538 for $\Y_{ES}$. This confirms that for networks with higher arc capacities, extreme points alone suffer a severe loss of optimality. Conversely, for binary problems such as the shortest path instances (\Cref{tab:metrics_side_sp_supp_v2}), this structural gap vanishes. Both sets achieve nearly identical near-optimal $\operatorname{RQR}^{\operatorname{CE}}$ values (0.763 vs. 0.762), making the smaller set $\Y_{ES}$ a perfectly sufficient candidate pool in binary contexts.

\paragraph{Performance for Small Cardinalities}
For the capacitated class N12, selecting just $k=5$ points from $\Y_S$ yields a mean $\operatorname{RQR}^{\operatorname{CE}}$ of 0.954. This means the restricted subset retains over 95\% of the optimal coverage quality achievable by the optimal subset of $\Y_N$. While representation quality naturally degrades as $k$ increases to 20, $\Y_S$ maintains robust relative performance. In class N07, the $\operatorname{RQR}^{\operatorname{CE}}$ for $\Y_S$ at $k=20$ remains at 0.753, whereas $\Y_{ES}$ drops to 0.484, effectively rendering the extreme set inadequate for larger representation sizes in capacitated problems. 

\paragraph{The Computational Trade-off}
The most significant advantage of this strategy is the massive reduction in computational overhead. As shown in \Cref{tab:instance_sizes_n12}, the mean cardinality of the full set $\Y_N$ for class N12 is roughly 1,895 points, while $\Y_S$ contains an average of only 143 points. By restricting the subset selection to $\Y_S$, the runtime of the filtering step is drastically reduced, since the cardinality of the candidate set is part of its input. Because the RQR values for small $k$ consistently exceed 0.90 across most metrics, this strategy proves to be a highly favorable trade-off: it guarantees near-optimal representations while requiring only a fraction of the computational cost associated with traditional a posteriori filtering.

\section{Conclusion}\label{chapt:concl}

 The key finding of this study is the critical role of supported non-dominated points as high-quality representations of the complete non-dominated point set in network optimization problems. 
 Across various problem classes and test instances, supported solutions consistently demonstrate high representational quality, as measured by several different quality indicators. For all instances, the hypervolume ratio of the supported non-dominated points as representations is always close to one and provides minimal coverage errors. Notably, for binary problems, the set of extreme supported points is nearly identical to the set of supported points, implying that $\Y_{ES}$ alone already provides a high-quality representation. In contrast, as arc capacities increase in network flow problems, the two sets become structurally different, and extreme supported non-dominated points yield significantly lower representation quality than the set of supported points. Our findings highlight that the structural properties of the network, such as sparsity, arc capacity, and cost correlation, strongly influence the representation quality, with $\Y_S$ adapting much better across diverse scenarios. Thus, supported non-dominated points provide a robust representation even in high-dimensional or densely structured MOIMCF problems. 
 
However, their cardinality can still be large and may constitute a significant fraction of the complete non-dominated point set, especially in MOIMCF problems. Therefore, we additionally propose a fixed-size representation approach, addressing the practical requirement of decision makers for small, manageable subsets of solutions (e.\,g., $k=5$ to $20$). In contrast to classical \emph{a posteriori} approaches, where optimal solutions to subset selection problems are derived from the complete non-dominated set, we restrict the candidate pool to the set of (extreme) supported non-dominated points. To quantify the resulting loss in quality, we introduce the relative quality ratio (RQR). Our results show that selecting representative subsets from $\Y_S$ leads to only negligible losses in quality compared to selecting from $\Y_N$. Remarkably, for capacitated problems, subsets derived from $\Y_S$ achieve RQR values exceeding 0.95 for small values of $k$, whereas subsets based on $\Y_{ES}$ perform considerably worse.

Our results support an efficient two-phase strategy for multi-objective network optimization. In the first phase, the supported non-dominated set is computed, avoiding the computationally intractable generation of the complete non-dominated set. In the second phase, $\Y_S$ is filtered to determine an optimal subset of $k$ representative points with respect to a given quality indicator. The resulting approach substantially reduces computational effort while maintaining a representation quality that is close to that obtained from the full non-dominated set.

Future research could derive criteria for identifying when supported non-dominated points are necessary for high-quality representations or whether extreme supported non-dominated points are sufficient. Additionally, given that $\mathcal{Y}_{S}$ can still be large in certain instances, developing specialized algorithms that directly identify $k$ supported points with high-quality without enumerating the entire $\mathcal{Y}_{S}$ set remains a promising direction for algorithmic development.

Moreover, it could be explored whether a strategically chosen subset of the supported non-dominated points, potentially derived from intermediate solutions, could provide sufficiently high-quality representations while reducing computational effort.

\paragraph{Acknowledgments} We sincerely thank Gonçalo Lopes for providing access to his open-source library and for his highly responsive, helpful technical support. The authors thankfully acknowledge the financial support of the Deutsche Forschungsgemeinschaft (DFG, German Research Foundation), project number 441310140.

\paragraph{Data Availability Statement}

The test instances and implementation are available in the associated repository~\citep{koenen2026supprepresentations}.

\paragraph{Conflicts of Interest}

The authors declare no conflicts of interest.

\providecommand{\natexlab}[1]{#1}
	\providecommand{\url}[1]{{#1}}
	\providecommand{\urlprefix}{URL }
	\expandafter\ifx\csname urlstyle\endcsname\relax
	\providecommand{\doi}[1]{DOI~\discretionary{}{}{}#1}\else
	\providecommand{\doi}{DOI~\discretionary{}{}{}\begingroup
		\urlstyle{rm}\Url}\fi
	\providecommand{\eprint}[2][]{\url{#2}}

\bibliographystyle{apalike}

\end{document}